\documentclass[journal]{IEEEtran}
\ifCLASSINFOpdf
\else
\fi
\usepackage[english]{babel}
\usepackage[utf8]{inputenc}
\usepackage[T1]{fontenc}
\usepackage{lmodern}
\usepackage{braket}
\usepackage{wrapfig}
\usepackage{graphicx}
\usepackage{amsmath,amssymb,amsfonts,verbatim}
\usepackage{hyperref}
\usepackage{mathtools}
\usepackage{enumitem}

\usepackage{xr}

\usepackage{amsmath,amssymb,amsfonts,verbatim}
\usepackage{algorithmic}
\usepackage{graphicx}
\usepackage[utf8]{inputenc}
\usepackage{textcomp}
\usepackage{xcolor}
\usepackage{braket}
\usepackage{soul}
\usepackage[colorinlistoftodos]{todonotes}

\usepackage{cite}
\usepackage{amsmath,amssymb,amsfonts,verbatim}
\usepackage{amsthm}
\usepackage{algorithmic}
\usepackage{graphicx}
\usepackage[utf8]{inputenc}
\usepackage{textcomp}
\usepackage{xcolor}
\usepackage{braket}
\usepackage{soul}
\usepackage{enumitem}
\usepackage[colorinlistoftodos]{todonotes}

\newtheorem{theorem}{Theorem}

\newtheorem{corollary}{Corollary}

\newcommand{\reals}{\mathbb{R}}
\newcommand{\CH}{\mathcal{H}}
\newcommand{\complex}{\mathbb{C}}

\newcommand{\CA}{\mathcal{A}}
\newcommand{\CX}{\mathcal{X}}

\newcommand{\CE}{\mathcal{E}}
\newcommand{\CLind}{\mathcal{L}_{(\alpha,\lambda,\phi)}}
\newcommand{\EX}{\mathbb{E}}
\newcommand{\Tr}{\operatorname{Tr}}
\newcommand{\PR}{\mathbb{P}}
\newcommand{\pdom}{D}
\newcommand{\pstate}{\rho_t}
\newcommand{\dcost}{\gamma}
\newcommand{\NumAct}{A}

\usepackage{tikz}
\usetikzlibrary{positioning}
 
\tikzstyle{pinstyle} = [pin edge={to-,thin,black}]
 
\allowdisplaybreaks
\begin{document}
%
\title{Quickest Detection for Human-Sensor Systems using Quantum Decision Theory}
%
%
%

\author{Luke~Snow,
        Vikram~Krishnamurthy,~\IEEEmembership{Fellow,~IEEE,}
        and~Brian~M.~Sadler,~\IEEEmembership{Life Fellow,~IEEE}
\thanks{Luke Snow and Vikram Krishnamurthy are with the Department
of Electrical and Computer Engineering, Cornell University, Ithaca,
NY, 14853 USA}
\thanks{Brian M. Sadler is with the U.S. Army Research Laboratory, Adelphi,
MD 20783 USA}

\thanks{This research was funded by National Science Foundation grant CCF-2112457,  Army Reesarch office grant W911NF-21-1-0093 , and Air Force Office of Scientific Research grant FA9550-22-1-0016.}}

\maketitle

\begin{abstract}
In mathematical psychology, recent models for human decision-making use Quantum Decision Theory to capture important human-centric features such as order effects and violation of the sure-thing principle (total probability law).
We construct and analyze a human-sensor system where a quickest detector aims to detect a change in an underlying state by observing human decisions that are influenced by the state.
   Apart from providing an analytical  framework for such human-sensor systems, we also analyze the structure of the quickest detection policy. We show that the quickest detection policy has a single threshold and the optimal cost incurred is lower bounded by that of the classical quickest detector. This indicates that intermediate human decisions strictly hinder detection performance. We also analyze the sensitivity of the quickest detection cost with respect to the quantum decision parameters of the human decision maker, revealing that the performance is robust to inaccurate knowledge of the decision-making process. Numerical results are provided which suggest that observing the decisions of more rational decision makers will improve the quickest detection performance. Finally, we illustrate a numerical implementation of this quickest detector in the context of the Prisoner's Dilemma problem, in which it has been observed that Quantum Decision Theory can uniquely model empirically tested violations of the sure-thing principle. 
\end{abstract}

\begin{IEEEkeywords}
Quickest Change Detection, Quantum Decision Making, Blackwell Dominance, Human-Sensor Interface
\end{IEEEkeywords}

%
\IEEEpeerreviewmaketitle

\section{Introduction}
%
%
%
%
In this paper we construct and analyze a sequential quickest detection framework which aims to detect a change in an underlying state by observing human decisions that  are influenced by the state. We incorporate a recently proposed human decision-making model from mathematical psychology which uses quantum probability to capture salient properties of human decision making which cannot be explained by classical expected utility or Markov models. Specifically, such quantum decision theories capture {\em order effects}
(humans  perceive $P(H|A\cap B) $ and $P(H|B \cap A)$ differently in decision making)
, violation of the {\em sure-thing principle} (human perception of probabilities in decision making violates the total probability rule), and temporal oscillations in decision preferences. We use the framework of \cite{martinez2016quantum} which models the human \textit{psychological state} as a time-evolving open-quantum system, which reaches a steady-state when deliberation has ceased.  

{\em Remark}. Quantum Decision Theory (QDT) models in  psychology do not claim that the brain is acting as a quantum device in any physical sense. Instead QDT  serves as a {\em parsimonious generative blackbox model} for human decision making that is backed up by experimental studies \cite{kvam2021temporal}, \cite{busemeyer2012quantum}.

The problem of 'quickest detection' \cite{shiryaev1963optimum} is fundamental to statistical signal processing \cite{TM10}, \cite{https://doi.org/10.48550/arxiv.2110.01581}, \cite{PH08}, and has applications in monitoring power networks \cite{chen2015quickest}, sensor networks \cite{raghavan2010quickest}, internet traffic \cite{lakhina2004diagnosing}, epidemic detection \cite{baron2004early}, genomic signal processing \cite{shen2012change}, seismology \cite{amorese2007applying}, and wireless communications \cite{lai2008quickest}. Quickest detection can be classified into \textit{non-parametric} and \textit{Bayesian} frameworks. Non-parametric approaches include the Cumulative Sum (CUSUM) \cite{page1954continuous} and Shiryaev-Roberts Procedure \cite{xie2021sequential}, which do not assume a \textit{prior} (distribution) on the change point time. Bayesian quickest detection utilizes a prior-posterior updating scheme and an assumed distribution for the change point \cite{xie2021sequential}.
In this paper, we consider \textit{Bayesian} quickest detection, in which the observed signals are \textit{human decisions} generated by a quantum decision maker. This problem of detecting a state change from the observation of human decisions is widespread, and includes contexts such as detecting a market shock by observing individual financial investment decisions, sentiment change through social media monitoring, or adversarial group strategy change through individual decision monitoring. 
We provide several structural results which characterize the optimal detection performance of the analyst who attempts to detect an underlying state change by observing human decisions. 

\begin{figure}[h!]
    \centering
    \includegraphics[height = 60mm]{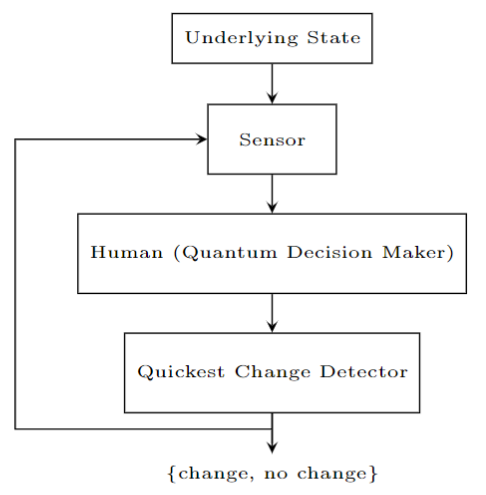}
    \caption{Sequential Quickest Change Detection with Human Decision Makers. Illustrated is the information flow for one time step $n$. This process repeats sequentially over discrete time $n=1,\dots$, with the state of nature jump changing at some unknown time, until the Quickest Detector declared 'change'.}
    \label{fig:Tikz_conc}
\end{figure}

\subsection{Human Sensor Based Change Point Detection}

The quickest detection framework of this paper is schematically illustrated in Fig.~\ref{fig:Tikz_conc}. An underlying state (e.g. asset value, etc.) changes at a geometrically distributed unknown time. At each time instant, a sensor obtains a noisy measurement of the underlying state (asset value, etc.), computes the posterior probability of the state, and provides this information to a human decision maker (e.g. as a recommendation). The human uses this information to choose an action at each time instant according to the quantum decision theory. These human decisions are monitored by a Quickest Change Detector, which computes a belief in the underlying state by exploiting knowledge of the quantum decision parameters. Based on the computed belief, the Quickest Detector then decides to continue or declares that a change has occurred, in which case the problem terminates.

\subsection{Context and Literature}
\subsubsection{Detection Theory for Human-Sensor Interaction}

The study of the interaction between sensor interfacing and human decision making demands utilization of tools from both statistical signal processing and behavioral economics or mathematical psychology, and specific examples which exploit models for human decision making can be found in robotics \cite{askarpour2019formal},  interactive marketing/advertising \cite{belanche2020consumer}, recommender systems \cite{lu2012recommender} and control of game-theoretic economic interactions \cite{9280374}.

.

One important  problem lying in this intersection is that of sequential change detection involving human decision makers, hereafter referred to as social sequential change detection. 
This problem has been studied previously (See \cite{krishnamurthy2021quickest} \cite{6205386}, \cite{6482232} and references therein) using models from behavioral economics and social learning. Recently, QDT  models for human decision making have been developed which account for a wide range of human decision making phenomena than traditional micro-economic models. 

\subsubsection{Quantum Decision Theory}
Generative models for human decision making are studied extensively in behavioral economics and psychology.
 The classical  formalisms  of human decision making are the Expected Utility models of Von-Neumann and Morgenstern (1953)\cite{morgenstern1953theory} and Savage (1954) \cite{savage1951theory}. Despite the successes of these models, numerous experimental findings, most notably those of Kahneman and Tverksy \cite{kahneman1982judgment}, have demonstrated violations of the proposed decision making axioms. There have since been subsequent efforts to develop axiomatic systems which encompass wider ranges of human behavior, such as the Prospect Theory \cite{kahneman2013prospect}.
Quantum Decision Theory (\cite{busemeyer2012quantum}, \cite{khrennikov2010ubiquitous}, \cite{yukalov2010mathematical} and references therein) has emerged as a new paradigm which is capable of generalizing current models and accounting for certain violations of axiomatic assumptions. For example, it has been empirically shown that humans routinely violate Savage's 'Sure-Thing Principle' \cite{khrennikov2009quantum}, \cite{aerts2011quantum}, which is equivalent to violation of the law of total probability, and that human decision making is affected by the order of presentation of information \cite{trueblood2011quantum} \cite{busemeyer2011quantum} ("order effects"). These violations are natural motivators for treating the decision making agent's mental state as a quantum state in Hilbert Space; The mathematics of quantum probability was developed as an explanation of observed self-interfering and non-commutative behaviors of physical systems, directly analogous to the findings which Quantum Decision Theory (QDT) aims to treat. Indeed, the models of Quantum Decision Theory have been shown to reliably account for violations of the 'Sure Thing Principle' and order effects \cite{busemeyer2012quantum}.
 

Within QDT, several recent advances have utilized quantum dynamical systems to model time-evolving decision preferences. The classical model for this type of time-evolving mental state is a Markovian model, but in \cite{busemeyer2009empirical} an alternative formulation based on Schr\"{o}dinger's Equation is developed. This model is shown to both reconcile observed violations of the law of total probability via quantum interference effects and model choice-induced preference changes via quantum projection. This is further advanced in \cite{asano2012quantum}, and  \cite{martinez2016quantum} where the mental state is modeled as an open-quantum system. This open-quantum system representation allows for a generalization of the widely used Markovian model of preference evolution, while maintaining these advantages of the quantum framework. Busemeyer et. al. \cite{kvam2021temporal} provide empirical analysis which supports the use of open-quantum models and conclude "An open system model that incorporates elements of both classical and quantum dynamics provides the 
best available single system account of these three characteristics—evolution, oscillation, and choice-induced 
preference change". 

The appeal of the quantum probabilistic model \cite{martinez2016quantum} is that it both provides a generalized decision making process which can account for certain empirically observed decision making phenomena
and it provides a quantitative way of reasoning about effects of cognitive biases and suboptimalities, such as bounded rationality, through the free-parametrization. We have also recently utilized this model within a human-machine assisted decision making scheme, in which a machine provides input signals to a human to dynamically steer the human's decisions towards optimality \cite{https://doi.org/10.48550/arxiv.2205.12378}. 

\subsection{Main Results and Outline}
In section \ref{sec:DQT} we outline the mathematical construction of the quantum decision making process \cite{martinez2016quantum}. Section \ref{sec:QCD} outlines the social sequential change detection protocol and the computation of the quickest detector's optimal policy, with the incorporation of the human decision making model of \cite{martinez2016quantum}. Sections \ref{sec:results} and \ref{sec:numerical} present our main theoretical results and computational validations regarding the quickest detector's performance and optimal policy resulting from our model. Specifically, we derive results for
    \begin{enumerate}
        \item \textit{Existence of a threshold optimal policy} (Theorem~\eqref{thm:sing_thres}): The quickest detector's optimal policy exhibits a single-threshold. This is in contrast to the multi-threshold optimal policy present in the social sequential change detection model of \cite{krishnamurthy2021quickest}. 
        \item \textit{Intermediate human decisions hinder detection performance} (Theorem~\eqref{thm:lower_bound}): Under the optimal policy, the quickest detector performs strictly worse in expectation than in the classical quickest change detection protocol. We note that this argument applies independently from the 'quantum' structure of the decision making process. However, this still can provide a useful lower bound on performance. 
        \item \textit{Sensitivity of detection performance to psychological parameters} (Theorem~\eqref{thm:estimate}): We provide an upper bound on the expected cost incurred by the quickest detector when only an estimate of the quantum decision maker's psychological parameters is available. This is useful as in any practical implementation one would work with a probabilistic estimate of these parameters.
        \item \textit{Detection performance depends on agent rationality} (Theorem \ref{thm:val_great}, Theorem \ref{thm:convex_dom}, computational results): We show that there exist disjoint convex regions in the psychological parameter space which induce performance ordering, i.e. the quickest detector performs strictly better when the decision maker has parameters in one region vs. the other. We provide a numerical simulation which validates this existence and suggests that the quickest detector performs better as the decision maker becomes more rational.
    \end{enumerate}
Along with these results, in Section~\ref{pris_dil} we provide a numerical example of the quickest detection scheme in the context of the Prisoner's Dilemma problem. The ability of the quantum model to account for empirically observed violations of the sure-thing principle \cite{busemeyer2006quantum} is illustrated in this context.

\section{Quantum Model for Human Decision Making}
\label{sec:DQT}

This  section presents the open-quantum system model that we will  use to represent the decision preference evolution of the human decision maker.
We define the evolution of the density operator of the decision maker using the open-system Quantum Lindbladian Equation, proposed in \cite{martinez2016quantum} and implemented in \cite{busemeyer2020application}. Reference~\cite{kvam2021temporal} provides empirical evidence which concludes that this open-system structure is the most parsimonious model which can capture observations of dynamical preference evolution such as oscillation and choice-induced preference change. This model provides a way of representing a dynamically evolving action preference distribution. However, for our purposes we abstract away from the time-evolution such that decisions are made from the \textit{steady-state distribution}, the existence of which is proved in \cite{martinez2016quantum}. This steady-state represents the ceasing of any deliberation. 

Readers who are unfamiliar with quantum probability may refer directly to the abstracted decision protocol in Section~\ref{dec_prot} . This will be used in subsequent sections for the observation likelihood in quickest detection. The details of Sections~\ref{psych_state} and~\ref{lindblad_construct} are not necessary for a high-level understanding of the quickest detection procedure, but provide insight into the novelty of this formulation and the impact of the psychological parameters.

 \subsection{Modeling  Psychological State via Quantum Probability}
 \label{psych_state}
 
 Suppose there are $n$ underlying states in the state space $\CX$, and $\NumAct$ actions in the action space $\CA$. For each state $i \in \{1,\dots,n\}$ construct a corresponding unit complex vector $\CE_i \in \complex^n$ such that $\{\CE_i\}_{i=1}^n$ are orthonormal. For each action $i \in \{1,\dots,\NumAct\}$, construct a complex vector $a_i \in \complex^{\NumAct}$ such that $\{a_i\}_{i=1}^{\NumAct}$ are orthonormal.  Denote $\CH_{\CX} = \textrm{span}\{\CE_1,\dots,\CE_n\}, \ \CH_{\CA} = \textrm{span}\{a_1,\dots,a_{\NumAct}\}$, and form the tensor product Hilbert space $\CH = \CH_{\CX} \otimes \CH_{\CA}$. 
The agent's psychological state is represented by a density operator $\pstate$ which acts on the Hilbert space $\pstate: \CH \to \CH$. Specifically, \[\pstate = \sum_{j} p_j \ket{\psi_j} \bra{\psi_j} \textrm{ with } \sum_j p_j = 1, \ \ket{\psi_j} \in \CH \ \forall j\] 
This construction is referred to as a \textit{mixed state} in quantum mechanics. A mixed state is a generalization of a \textit{pure state} to a probability distribution over pure states. We use a mixed state representation (density operator) for the sake of generality. The psychological state $\pstate$ evolves according to the Lindbladian operator $\CLind$ by the ordinary differential equation\footnote{The reader may be familiar with the Schr\"{o}dinger equation which governs the time evolution of \textit{closed} quantum systems. The Lindbladian equation is a generalization which governs the time evolution of \textit{open} quantum systems (i.e. those that interact with an external dissipative environment). The recent literature in psychology uses the Linbladian framework to model human decision making.}
\begin{equation}
    \label{Lind_ev}
    \frac{d\pstate}{dt} =  \CLind \pstate
\end{equation} where $(\alpha,\lambda,\phi)$ are free parameters which govern the evolution, each having a psychological interpretation, see \cite{martinez2016quantum}. Implicit in $\CLind$ is a belief $\eta(x)$ in the underlying state $x \in \{1,\dots,n\}$ and a utility function $u: \CA \times \CX \to \reals$. The psychological state $\pstate$ encodes a time dependent probability distribution $\Gamma(a,t)$ over actions $a \in \CA$ in the following way. Let $P_i$ be the projector on to the subspace spanned by action vector $a_i \in \CH$, then $\Gamma(a_i,t) = \Tr(P_i \pstate P_i^{\dagger})$, where $P_i^{\dagger}$ is the adjoint of $P_i$.

\subsection{Lindbladian Operator Construction}
\label{lindblad_construct}
The evolution of the density operator is given by $
\frac{d\pstate}{dt} = \mathcal{L}_{(\alpha,\lambda,\phi)} \, \pstate
$ where
\begin{equation}
\label{Lindblad}
\begin{split}
     &\mathcal{L}_{(\alpha,\lambda,\phi)} \, \pstate =-i(1-\alpha)[H,\pstate]\\&+\alpha \sum_{m,n}\gamma_{(m,n)}\left(L_{(m,n)}\,\pstate\, L_{(m,n)}^{\dagger}-\frac{1}{2}\{L_{(m,n)}^{\dagger}L_{(m,n)},\pstate\}\right)    
\end{split}
\end{equation}
Here $[A,B] = AB - BA$, $\{A,B\} = AB + BA$, $A^*$ is complex conjugate of $A$,
${H}=\textrm{diag}({1}_m,\cdots,{1}_m)_{mn\times\,mn}$ with ${1}_m$ an $m \times m$ matrix of ones and $L_{(m,n)} = \ket{m}\bra{n}$, where $\ket{k}$ is the $k$'th basis vector of $\mathcal{H}$. The coefficient $\gamma_{(m,n)}$ is given by the $(m,n)$'th element of the \textit{cognitive matrix} $C(\lambda,\phi)$:
\begin{equation}
    \label{gamma_mn}
    \gamma_{(m,n)}:=[C(\lambda,\phi)]_{m,n}=[(1-\phi)\Pi^{T}(\lambda) + \phi B^{T}]_{m,n}
\end{equation}
For utility function $u: \CA \times \CX \rightarrow \reals$, construct
\begin{equation}
\label{subj_util}
    p(a_j|\mathcal{E}_l)=\frac{u(a_j|\mathcal{E}_l)^{\lambda}}{\sum_{j=l}^{\NumAct} u(a_j|\mathcal{E}_l)^{\lambda}}
\end{equation}
and define
\begin{equation} \label{Pi_mat}
    \begin{split}
{P}(\mathcal{E}_l)&:=\begin{bmatrix}
    p(a_1|\mathcal{E}_l) & p(a_2|\mathcal{E}_l)&\cdots& p(a_{m}|\mathcal{E}_l) \\
\end{bmatrix}\otimes{1}_{n\times1} \\
    \Pi(\lambda) &= \textrm{diag}(
        P(\mathcal{E}_1),\cdots,P(\mathcal{E}_{n}))
        \end{split}
\end{equation}
where ${1}_{n\times1}$ is a vector with all 1's and, $A \otimes B$ is the kronecker product of $A$ and $B$. Define
$\eta_k(s) = p(s | u_k, y_k)$
given the noisy observation $y_k$ and input signal $u_k$, with $s \in \CX$.
We define
\begin{equation}
    \label{B_Matrix}
    B:={\begin{bmatrix}
    \eta_k(\mathcal{E}_1) & \eta_k(\mathcal{E}_2)&\cdots&\eta_k(\mathcal{E}_{n}) 
  \end{bmatrix} }\otimes{1}_{m\times1}\otimes \mathbb{I}_{m\times m}
\end{equation}
See \cite{martinez2016quantum} for the psychological motivation behind this structure. \eqref{Lindblad} is the standard form of the Lindblad-Kossakowski ordinary differential equation, which governs the behavior of quantum systems interacting with an external environment, or 'open' quantum systems. 

\subsection{Practicality in Human Decision Making}
\label{practicality}
The above quantum model for human decision making accounts for violations of the sure-thing principle (STP), which we now describe. Suppose there exists an action \textit{a} and two states $\CE_1,\CE_2$. Suppose $\Gamma$ is a non-degenerate posterior belief (strictly in the interior of the unit simplex) of the underlying state. The violation of the sure thing principle occurs when $P(\textit{a}|\Gamma)$ is not a convex combination of $P(\textit{a}|\CE_1)$ and $P(\textit{a}|\CE_2)$, i.e. 
\[P(\textit{a}|\Gamma) \neq \epsilon\,P(\textit{a}|\CE_1) + (1-\epsilon)\,P(\textit{a}|\,\CE_2) \ \forall\, \epsilon \in (0,1)\]
Pothos and Busemeyer \cite{pothos2009quantum} (see also \cite{khrennikov2009quantum}) review  empirical evidence for the violation of STP and show how quantum models can account for it by introducing quantum interference in the probability evolution. Note that this violation cannot be accounted for by traditional models which rely on classical probability, as the sure-thing principle follows directly as a consequence of the law of total probability. 

The parameters $(\alpha,\lambda,\phi)$ also allow for practical psychological interpretation. The parameter $\alpha$ interpolates between the purely quantum preference evolution and the dissipative Markovian evolution in \eqref{Lindblad}, and thus a higher $\alpha$ corresponds to increased \textit{rationality}, in the sense of choosing actions which accord with classical expected utility maximization. $\lambda$ is a measure for \textit{bounded rationality}, as (from \eqref{subj_util}) it is a monotonic measure of the ability to discriminate between the profitability of different options. The interpretation of $\phi$ (in \eqref{gamma_mn}) is more nuanced, but can be thought of as the relevance of the formation of a belief in the underlying state to the decision making process. See \cite{martinez2016quantum}, \cite{kvam2021temporal} for  detailed discussion on these interpretations. 
  
\subsection{Decision making protocol}
\label{dec_prot}
Each quantum decision maker (human) in the sequential decision process behaves as follows. The agent has initial psychological state $\rho_0$ and utility $u: \CA \times\CX \to \reals$. An underlying state distribution $\eta(x)$ is provided by a Bayesian inference machine (Sensor). $u$ and $\eta(x)$ parameterize $\CLind$, along with psychological parameters $\alpha,\lambda,\phi$. The psychological state at time $t$, $\pstate$, evolves according to \eqref{Lind_ev} and induces a distribution $\Gamma^{\eta}(a,t)$ over the action space as 
\begin{equation}
\label{Gamma}
    \Gamma^{\eta}(a,t) = \Tr(P_a\pstate P_a^{\dagger})
\end{equation} 
By \cite{martinez2016quantum}, we are guaranteed the existence of a  \textit{steady-state distribution}
\[\Gamma^{\eta}(a) = \lim_{t \to \infty}\Gamma^{\eta}(a,t)\]
We assume action $a_n$ is taken probabilistically according to the steady-state distribution $\Gamma^{\eta}(a)$ which is independent from the initial state $\rho_0$. This represents the action choice occurring after deliberation has ended, and the steady state is typically reached relatively quickly\footnote{See \cite{martinez2016quantum} for a proof of the steady state and a discussion of relaxation times of this evolution}. We can then abstract away from the time dependence to get the map
\begin{equation}
\label{act_map}
    \CLind : (\eta(x),u(x,a)) \to \Gamma^{\eta}(a)
\end{equation}
At each discrete time point of the quickest detection protocol, the agent:
\begin{itemize}
    \item consists of initial psychological state $\rho_0$, utility $u: \CA \times \CX \to \reals$, and parametrization $(\alpha,\lambda,\phi)$. Note that these quantities are time independent and thus constant for all discrete time steps.
    \item is provided state information in the form of a Bayesian posterior $\eta_n(x)$ by the Sensor.
    \item deliberates until reaching a steady-state action distribution $\Gamma^{\eta_n}(a)$, from map \eqref{act_map}.
    \item takes action $a_n$ probabilistically from $\Gamma^{\eta_n}(a)$
\end{itemize}

\subsection{Summary} The psychological state is represented as a density operator $\pstate$ acting on the Hilbert space $\CH$, which is formed as a tensor product of vector spaces spanned by orthonormal state and action vectors. This representation of the psychological state encodes quantum 'amplitudes' over joint state-action pairs. The psychological state $\pstate$ evolves according to \eqref{Lind_ev}, where $\CLind$ is the quantum Lindbladian operator, and is constructed in a specific way \cite{martinez2016quantum} to reflect a psychological preference evaluation process. The psychological state evolves until it reaches a steady-state, corresponding to a halting of any further deliberation, and at which point a decision is taken probabilistically. Thus, we can abstract away from the time evolution and represent the decision making process by the map \eqref{act_map}. This 'quantum' psychological preference evolution acts as black-box model which generalizes the analogous classical Markovian preference evolution model.

\section{Quickest Change Detection with Quantum Decision Maker}
\label{sec:QCD}

We now introduce the quickest change detection protocol and the formulation of an optimal policy for such a protocol.  
The aim of quickest detection is to determine the jump time $\tau^0$ of the state of nature $\{x_n\}$ i.e., evaluate the optimal stationary policy $\mu^*$ of the global decision maker that minimizes the Kolmogorov-Shiryaev criterion for detection of disorder:
\begin{align}
\begin{split}
\label{opt_pol}
    &J_{\mu^*}(\pi) = \inf_{\mu} J_{\mu}(\pi), \\
    &J_{\mu}(\pi) = d\EX_{\mu}[(\tau-\tau^0)^+] + f\PR_{\mu}(\tau < \tau^0)
    \end{split}
\end{align}
where $\tau = \inf\{n : u_n = 1 \}$ is the time at which the global decision maker announces the change. The parameters $d$ and $f$ specify the delay penalty and false alarm penalty, respectively. 

The optimal policy $\mu^{*}(\pi)$ \eqref{opt_pol} can be formulated as the solution of a stochastic dynamic programming equation. The quickest detection problem \eqref{opt_pol} is an example of a stopping-time partially observed Markov decision process (POMDP) with a stationary optimal policy.

We now introduce some notation, then describe the protocol in detail.
\begin{enumerate}[label=\roman*)]
    \item The state of nature $\{x_n \in \{1,2\}, n \geq 0\}$ models the change event which we aim to detect. $x_n$ starts in state $2$ and jumps to state $1$ at a geometrically distributed random time $\tau^0$ with $\EX[\tau^0] = \frac{1}{1-p}$ for some $p \in [0,1)$. So, $\{x_n\}$ is a 2-state Markov chain with absorbing transition matrix and initial probability
    \begin{equation}
    \label{trans_mat}
        P =  \left[\begin{matrix}  
        1 & 0 \\  
        1-p & p   
        \end{matrix} \right] 
        , \ \pi_0 = \left[\begin{matrix}  
        0 \\  
        1   
        \end{matrix} \right] 
    \end{equation}
    with change time $\tau^0 = \inf\{n: x_n=1\}$.
    \item The quantum decision agents act sequentially. A sensor observe the state of nature $x_n$ in noise and computes a Bayesian posterior distribution $\eta(x)$ of the underlying state. This is given to the human, who then makes a local decision $a_n$ according to the steady-state action distribution $\Gamma^{\eta}(a_i)$ induced by the Lindbladian operator $\mathcal{L}_{(\alpha,\lambda,\phi,\eta)}$ and map \eqref{act_map}.
    \item Based on the history of local actions $a_1,\dots,a_n$, the global decision maker chooses action 
    \[u_n = \{1 (\textrm{stop and announce change}), \ 2(\textrm{continue}) \} \]
    \item Define the public belief $\pi_n$ and private belief $\eta_n$ at time $n$ as the posterior distributions initialized with $\eta_0 = \pi_0 = [0,1]':$
    \begin{align}
        \begin{split}
            &\pi_n(x) = \PR(x_n = x | a_1,\dots,a_n), x = 1,2 \\
            &\eta_n(x) = \PR(x_n = x | a_1,\dots,a_{n-1},y_n),
        \end{split}
    \end{align}
    where $y_n$ is the private observation recorded by agent $n$. We have $\pi_n(x), \eta_n(x) \in \Pi$, the unit one-simplex.
    \end{enumerate}

\begin{figure}[h!]
     \centering
     \includegraphics[width=1\linewidth,scale=0.5]{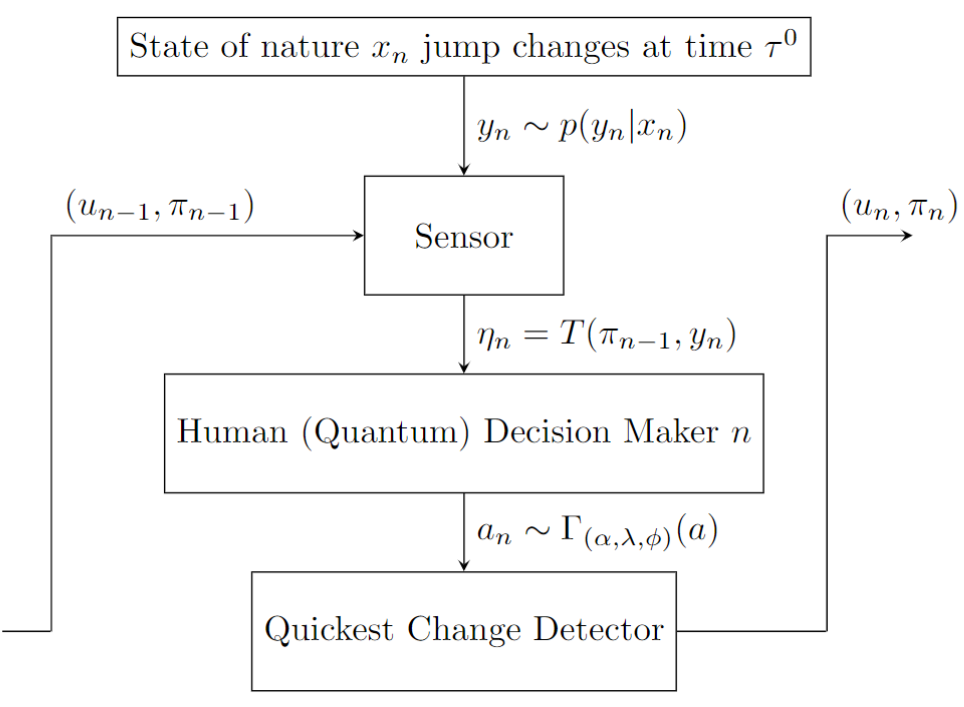}
     \caption{Sequential Quickest Change Detection with quantum agents. The underlying state of nature $x_n$ jump changes at time $\tau^0 \sim$Geo(1-$p$), where $p$ is known. At each time $n$ a sensor observes $y_n \sim P(y|x)$, and the public belief signal $\pi_{n-1}$ from the previous time point. The sensor outputs a private belief $\eta_n$ (obtained via Bayesian update) in the underlying state to the quantum decision maker. The decision maker's parameterized psychological Lindbladian $\CLind(\eta_n)$ evolves to steady state $\Gamma$ \eqref{act_map} and an action $a_n$ is taken probabilistically from $\Gamma$. The quickest detector sees $a_n$ and outputs its public belief $\pi_n$ and signal $u_n$ according to \eqref{glob_dec_mak} and \eqref{glob_act}.}
     \label{fig:QCD_Proc}
 \end{figure}

\subsection{Change Detection Protocol \cite{krishnamurthy2021quickest}}
\label{QCD_prot}
We now detail the multi-agent quickest change detection protocol:
\begin{enumerate}
    \item \textit{Local human decision maker n}
    \begin{enumerate}[label=\roman*)]
        \item Sensor obtains public belief $\pi_{n-1}$ and signal $u_{n-1}$ from global decision maker. The process only continues if $u_{n-1} = 2$.
        
        \item Let $\mathcal{Y}$ denote the observation space. The sensor records noisy observation $y_n \in \mathcal{Y}$ of state of nature $x_n$ with conditional density 
        \begin{equation}
            B_{x,y} = p(y_n = y|x_n = x)
        \end{equation}
        
        \item \textit{Private Belief}. The sensor evaluates the Bayesian private belief 
        \begin{align}
        \label{priv_bel}
        &\eta_n = T(\pi_{n-1},y_n), \ T(\pi,y) = \frac{B_y P' \pi}{\sigma(\pi,y)}, \\
        & \sigma(\pi,y) = \boldsymbol{1}'B_yP'\pi, \ B_y = \textrm{diag}(B_{1,y}B_{2,y})
        \end{align}
        and feeds this to the human agent. 
        
        \item \textit{Local decision}. The agent's private belief $\eta_n$ parameterizes the Lindbladian operator $\CLind$. This induces a steady-state action probability distribution $\Gamma^{\eta}(a_i)$ via the map \eqref{act_map}, and the action $a_n$ is taken probabilistically according to $\Gamma^{\eta}$.
    \end{enumerate}
    \item \textit{Quickest Detector}. Based on the decisions $a_n$ of local decision maker $n$, the quickest detector:
    \begin{enumerate}[label=\roman*)]
        \item Updates the public belief from $\pi_{n-1}$ to $\pi_n$ as 
        \begin{align}
            \begin{split}
            \label{glob_dec_mak}
                &\pi_n = \bar{T}(\pi_{n-1},a_n) \\
                &\bar{T}(\pi,a) = \frac{R_{\pi}(a)P'\pi}{\bar{\sigma}(\pi,a)}, \ \bar{\sigma}(\pi,a) = \boldsymbol{1}'R_{\pi}(a)P'\pi \\
                &R_{\pi}(a) = \textrm{diag}(R_{1,\pi}(a), R_{2,\pi}(a)) \\
                &R_{x,\pi}(a_n) = \PR(a_n = a | x_n = x, \pi_{n-1})
            \end{split}
        \end{align}
        The action probabilities $R_{x,\pi}(a)$ are computed as 
        \begin{equation}
        \label{act_prob}
            R_{x,\pi}(a) = \int_{\mathcal{Y}}\bar{\Gamma}_{y}^{\pi_{n-1}}(a)B_{x,y}dy
        \end{equation}
        where $\bar{\Gamma}_{y}^{\pi_{n-1}}$ is the QDM's induced action distribution \eqref{act_map} given public belief $\pi_{n-1}$ and observation $y$. Specifically, $\bar{\Gamma}_{y}^{\pi_{n-1}}$ is the output of the map \eqref{act_map}, with input $\eta(x) = T(\pi_{n-1},y), u(x,a)$ and \textit{estimated} parametrization $(\hat{\alpha},\hat{\lambda},\hat{\phi})$. Observe that here the quickest detector has an estimate of the psychological parametrization $(\alpha,\lambda,\phi)$; later we will investigate the performance sensitivity to this estimate.
        \item Chooses global action $u_n$ using optimal policy $\mu^*$:
        \begin{equation}
        \label{glob_act}
            u_n = \mu^*(\pi_n) \in \{1 (\textrm{stop}), \ 2 (\textrm{continue})\}.
        \end{equation}
        \item Is $u_n = 2$, then set $n$ to $n+1$ and go to step 1. If $u_n = 1$, then stop and announce change. 
    \end{enumerate}
\end{enumerate}

We assume the global decision maker knows $P$ \eqref{trans_mat} and the agent's action $a_n$, and has an estimate of the agent's psychological dynamics $\mathcal{L}_{(\hat{\alpha},\hat{\lambda},\hat{\phi})}$. The global decision maker does not know the observation $y_n$ or the private belief $\eta_n$. For simplicity, we assume all agents have the same psychological dynamics 
$\CLind$ (Such as if the same agent acts sequentially), otherwise the optimal detection strategy is non-stationary.
The update \eqref{glob_dec_mak} is where the quantum decision theory enters our quickest detection formulation. In simple terms, the action of the human is a probablistic function of the noisy measurement of the sensor. So the likelihood of the action given the state enters our computation for the belief state in quickest detection.

\subsection{Quickest Detector Optimal Policy \cite{krishnamurthy2021quickest}}
\label{QCD_OptPol}
Considering the aim of quickest detection, characterized by \eqref{opt_pol}, we now outline the details of the optimal policy stochastic dynamic programming formulation. 

\begin{enumerate}
    \item \textit{Costs}: To present the dynamic programming equation we first formulate the false alarm and delay costs \eqref{opt_pol} incurred by the global decision maker in terms of the public belief.
    \begin{enumerate}[label=\roman*)]
        \item \textit{False alarm penalty}: If global decision $u_n = 1$ (stop) is chosen before the change point $\tau^0$, then a false alarm penalty is incurred. The false alarm event $\{x_n = 2, u_n = 1\}$ represents the event that a change is announced before the change happens at time $\tau^0$. Recall \eqref{trans_mat} the jump change occurs at time $\tau^0$ from state 2 to state 1. Then recalling $f \geq  0$ is the false alarm penalty in \eqref{opt_pol}, the expected false alarm penalty is 
        \begin{align}
        \begin{split}
            &f\PR_{\mu}(\tau < \tau^0) = f\EX_{\mu}\{\EX [I(x_n = 2, u_n = 1) | \mathcal{G}_n]\} \\
            &\mathcal{G}_n = \sigma\textrm{-algebra generated by }(a_1,\dots,a_n)
        \end{split}
        \end{align}
        Clearly $\EX [I(x_n = 2, u_n = 1) | \mathcal{G}_n]$ can be expressed in terms of the public belief $\pi_n(2) = P(x_n = 2|a_1,\dots,a_n)$ as 
        \begin{equation}
        \label{fa}
            C(\pi_n,u_n=1) = f e_2'\pi_n, \ \ \textrm{where } e_2 = [0,1]'
        \end{equation}
        
        \item \textit{Delay cost of continuing}: If global decision $u_n = 2$ is taken then Protocol 1 continues to the next time. A delay cost is incurred when the event $\{x_n = 1,u_n = 2\}$ occurs, i.e. no change is declared at time $n$. The expected delay cost is $d\EX[I(X_n = 1, u_n = 2) | \mathcal{G}_n]$ where $d > 0$ denotes the delay cost. In terms of the public belief, the delay cost is 
        \begin{equation}
        \label{dc}
            C(\pi_n, u_n = 2) = de_1'\pi_n, \ \textrm{where }e_1 = [1,0]'
        \end{equation}
        
        We can re-express Kolmogorov-Shiryaev criterion \eqref{opt_pol} as 
        \begin{equation}
        \label{E_Cost}
            J_{\mu} = \EX_{\mu}\{\sum_{n=0}^{\tau-1} C(\pi_n,2) + C(\pi_{\tau},1) \}
        \end{equation}
        where $\tau = \inf\{n : u_n = 1\}$ is adapted to the $\sigma\textrm{-algebra} \ \mathcal{G}_n$. Since $C(\pi,1),C(\pi,2)$ are non-negative and bounded for all $\pi \in \Pi$, stopping is guaranteed in finite time.  
    \end{enumerate}
    
    \item \textit{Bellman's equation for Quickest Detection Policy}: Consider the costs \eqref{fa}, \eqref{dc} defined in terms of the public belief $\pi$. Then the optimal stationary policy $\mu^{*}(\pi)$ defined in \eqref{opt_pol} and associated value function $V(\pi)$ are the solution of Bellman's dynamic programming functional equation
    \begin{align}
        \begin{split}
        \label{bellman}
            &Q(\pi,1) := C(\pi,1) \\
            &Q(\pi,2) := C(\pi,2) + \sum_{a \in \mathcal{A}_1 \times \mathcal{A}_2} \mathcal{V}(\bar{T}(\pi,a))\bar{\sigma}(\pi,a) \\
            &\mu^*(\pi) = \textrm{argmin}\{Q(\pi,1),Q(\pi,2)\}, \\ &\mathcal{V}(\pi) = \min\{Q(\pi,1),Q(\pi,2)\} = J_{\mu}^*(\pi)
        \end{split}
    \end{align}
    The public belief update $\bar{T}$ and normalization measure $\bar{\sigma}$ were defined in \eqref{glob_dec_mak}. The goal of the global decision-maker is to solve for the optimal quickest change policy $\mu^*$ in \eqref{bellman} or, equivalently, determine the optimal stopping set $\mathcal{S}$
    \begin{equation}
        \mathcal{S} = \{\pi: \mu^*(\pi) = 1 \} = \{\pi: Q(\pi,1) \leq Q(\pi,2) \}
    \end{equation}
    
    \item \textit{Value Iteration Algorithm}: The optimal policy $\mu^*(\pi)$ and value function $\mathcal{V}(\pi)$ can be constructed as the solution of a fixed point iteration of Bellman's equation \eqref{bellman}. The resulting algorithm is called the value iteration algorithm. The value iteration algorithm proceeds as follows: Initialize $\mathcal{V}_0(\pi) = 0$ and for iterations $k=1,2,\dots$
    \begin{align}
        \begin{split}
        \label{value_itr}
            &\mathcal{V}_{k+1}(\pi) = \min_{u \in \mathcal{U}}Q_{k+1}(\pi,u), \\
            &\mu_{k+1}^*(\pi) = \textrm{argmin}_{u \in \mathcal{U}}Q_{k+1}(\pi,u), \ \pi \in \Pi, \\
            &Q_{k+1}(\pi,1) = C(\pi,1), \\ &Q_{k+1}(\pi,2) = C(\pi,2) + \sum_{a \in \mathcal{A}_1 \times \mathcal{A}_2} \mathcal{V}_k (\bar{T}(\pi,a))\bar{\sigma}(\pi,a)
        \end{split}
    \end{align}
    Let $\mathcal{B}$ denote the set of bounded real-valued functions on $\Pi$. For any $\mathcal{V},\Tilde{\mathcal{V}} \in \mathcal{B}$ and $\pi \in \Pi$, define the sup-norm metric sup$\parallel \mathcal{V}(\pi) - \Tilde{\mathcal{\pi}}\parallel$. Since $C(\pi,1),C(\pi,2),\pi \in \Pi$ are bounded, the value iteration algorithm \eqref{value_itr} generates a sequence of lower semi-continuous value functions $\{\mathcal{V}_k\} \subset \mathcal{B}$ that converges pointwise as $k \to \infty $ to $\mathcal{V}(\pi) \in \mathcal{B}$, the solution of Bellman's equation.
\end{enumerate}

\section{Characterizing the Structure of the Quickest Detector}
\label{sec:results}
In this section we analyze several structural properties of the quickest detection protocol detailed in Sec. \ref{sec:QCD}. Our results in this section are structured as follows: In Section \ref{sing_thres} we prove that the optimal policy \eqref{bellman} has a single threshold structure. In Section \ref{lb_perf} we provide a lower bound on the optimal cost incurred by the quickest detector via the policy of Sec. \ref{sec:QCD}. Specifically, this lower bound is given by the optimal cost incurred within the classical quickest change detection protocol, i.e. without intermediate human decisions. The key idea here is to use Blackwell dominance between matrices characterizing the quickest detector observations and the noisy sensor observations. In Section \ref{estimate} we consider the performance sensitivity to the quickest detector's estimate of the psychological parameterization, and prove an upper bound on the cumulative cost incurred in terms of the cumulative cost incurred given perfect knowledge of the parameterization and a KL Divergence term. 

\subsection{Existence of a Threshold Optimal Policy}
\label{sing_thres}
We will show that, given the quantum decision making quickest change detection protocol detailed in Section \ref{QCD_prot}, the quickest detector's optimal policy \eqref{opt_pol} exhibits a single-threshold behavior.

\begin{theorem}
\label{thm:sing_thres}
Given the quantum decision making quickest change detection protocol detailed in Section \ref{QCD_prot}, the quickest detector's optimal policy $\mu^*$ \eqref{opt_pol} exhibits a single threshold state $\pi'$ such that
\[\mu^*(\pi) = \begin{cases} 2, \ \  \pi < \pi' \\  1, \ \ \pi \geq \pi'\end{cases} \]
\end{theorem}
\begin{proof}
See Appendix \ref{pf:sing_thres}.
\end{proof}

In Section~\ref{pris_dil} we numerically implement the value iteration algorithm~\eqref{value_itr} in the context of a 'Prisoner's Dilemma' quickest detection scheme. In particular,  Fig.~\ref{fig:PD_OptPol} demonstrates the single threshold behavior of the optimal policy $\mu^*(\pi)$ \eqref{bellman}. 

This optimal policy structure is in contrast to the multi-threshold policy obtained in \cite{krishnamurthy2021quickest}, in which an \textit{anticipatory} model was used for the human decision makers. Within a multi-threshold (non-convex stopping region) policy, there exist points where the optimal behavior is to transition from declaring change to declaring no change as the \textit{probability of change increases}. This is not only counterintuitive, but makes the design of human-sensor quickest detectors more complex. Thus the single threshold policy exhibited in our case is desirable for intuitive and practical design purposes.

\subsection{Lower bound for performance} 
\label{lb_perf} 
We now show that the optimal cost incurred by quickest change detection with quantum agents is greater than that incurred by the classical Bayesian framework. We note that this result is not due to the 'quantum' behavior, but holds because of the general local-global decision maker setup. Nevertheless, this is useful since performance analysis of standard quickest detection \cite{tartakovsky2005general} applies as a lower bound for quickest detection with quantum agents. Consider the optimal policy and cost of the \textit{classical} Bayesian quickest change detection.  \cite{tartakovsky2005general}. Similar to \eqref{value_itr}, the optimal policy $\underline{\mu}^*(\pi)$ and cost $\underline{\mathcal{V}}(\pi)$ incurred by the classical quickest detection, satisfy the stochastic dynamic programming equation: 
\begin{align}
    \begin{split}
    \label{clas_val_it}
        &\underline{\mu}^*(\pi) = \textrm{argmin}_{u \in \mathcal{U}}\underline{Q}(\pi,u), \ \ \  \underline{\mathcal{V}}(\pi) = \min_{u \in \mathcal{U}}\underline{Q}(\pi,u), \ \\
        &\textrm{where }\underline{Q}(\pi,2) = C(\pi,2) + \sum_{y \in \mathcal{Y}}\underline{\mathcal{V}}(T(\pi,y))\sigma(\pi,y), \\
        &\underline{Q}(\pi,1) = C(\pi,1), \ \ \ \underline{J}_{\mu^*}(\pi) = \underline{\mathcal{V}}(\pi)
    \end{split}
\end{align}
Here $T(\pi,y)$ is the Bayesian filter update defined in \eqref{priv_bel} and $\underline{J}_{\mu^*}(\pi)$ is the cumulative cost of the optimal policy starting with initial belief $\pi$. Note that in classical quickest detection, there is no public belief update \eqref{glob_dec_mak} or interaction between public and private beliefs. 

\begin{theorem}
\label{thm:lower_bound}
Consider the quantum decision making quickest change detection protocol in Section \ref{QCD_prot} and the associated value function $\mathcal{V}(\pi)$ in \eqref{value_itr}. Consider also the classical quickest change detection problem with value function $\underline{\mathcal{V}}(\pi)$ in \eqref{clas_val_it}. Then for any initial belief $\pi \in \Pi$, the optimal cost incurred by the classical quickest detection is smaller than that of quickest detection with quantum decision agents. That is, $\underline{\mathcal{V}}(\pi) \leq \mathcal{V}(\pi) \ \forall \pi \in \Pi$.
\end{theorem}
\begin{proof}
See Appendix \ref{pf:lower_bound}.
\end{proof}

Informally, this result can be interpreted by the observation that the intermediate human decision making process results in loss of information pertaining to the underlying state. Indeed, we use Blackwell Dominance arguments within the proof, which formalize this notion of cascaded information loss. The practical interpretation is that regardless of the human psychological parametrization, i.e. perfectly rational etc., the hierarchical detection structure in which there is an intermediate human decision making process results in decreased detection performance (by way of \textit{increased cost} through the value function $\mathcal{V}(\pi)$).

\subsection{Sensitivity of Detection Performance to Psychological Parameters}
\label{estimate}
Recall that the quickest detector uses an estimate of the psychological parameters $(\alpha,\lambda,\phi)$. Thus, we would like to characterize how the quickest detection performance depends on such an estimate. In this section, we quantify this question and provide a bound on the deviation of the performance from that incurred by perfect knowledge of the psychological parameters. 

First we begin by defining some notation. Recall the domain of the psychological parameters (for brevity we denote this $\pdom$ \[(\alpha,\lambda,\phi) \in \pdom = [0,1]\times[0,\infty)\times[0,1] \subset \reals^3\]
Thus define $\Lambda$ to be a probability density function in $\pdom$ 
\[\Lambda: \pdom \to [0,1],\ \int_{\pdom} \Lambda(\gamma) d\gamma = 1 \]
representing the quickest detector's probabilistic estimate of the local decision maker's psychological parameters. Denote the actual psychological parameterization of the local decision maker by $\bar{\gamma} \in \pdom$. 

We now reconsider the decision making protocol from Section \ref{QCD_prot} when the quickest detector only has this probabilistic estimate of the local decision maker's psychological parameters. The \textit{Local quantum decision maker n} step remains the same, except let us now denote the steady-state action distribution by $\Gamma_{\bar{\gamma}}(a_i)$ to denote that this is a result of the true parameterization $\bar{\gamma}$. The global decision maker now updates the public belief from $\pi_{n-1}$ to $\pi_n$ as
\begin{align}
    \begin{split}
    \label{glob_dec_mak2}
        &\pi_n = \hat{T}(\pi_{n-1},a_n) \\
        &\hat{T}(\pi,a) = \frac{\hat{R}_{\pi}(a)P'\pi}{\bar{\sigma}(\pi,a)}, \ \bar{\sigma}(\pi,a) = \boldsymbol{1}'\hat{R}_{\pi}(a)P'\pi \\
        &\hat{R}_{\pi}(a) = \textrm{diag}(\hat{R}_{1,\pi}(a), \hat{R}_{2,\pi}(a)) \\
        &\hat{R}_{x,\pi}(a_n) = \PR(a_n = a | x_n = x, \pi_{n-1}, \Lambda)
    \end{split}
\end{align}
The action probabilities $\hat{R}_{x,\pi}(a)$ are now computed as 
\begin{equation}
\label{act_prob2}
    \hat{R}_{x,\pi}(a) = \int_{\mathcal{Y}}\int_{\pdom}\bar{\Gamma}_{y,\gamma}^{\pi_{n-1}}(a)\Lambda(\gamma)B_{x,y}d\gamma dy
\end{equation}
where $\bar{\Gamma}_{y,\gamma}^{\pi_{n-1}}$ is the QDM's induced action distribution \eqref{act_map} given public belief $\pi_{n-1}$, observation $y$, and psychological parameterization $\gamma \in \pdom$. The quickest detector then chooses action $u_n$ according to \eqref{glob_act}, where now the optimal policy is computed using the value iteration algorithm \eqref{value_itr} with this new function $\hat{T}(\pi,a)$.

We are now interested in characterizing how this generalized procedure effects the quickest change performance. We now define some notation which will allow us to reason about this.
Notice that quickest change decision making protocol in Section \ref{QCD_prot} is completely characterized as a two-state partially Observed Markov Decision Process (POMDP) with underlying state transition matrix $P$ and observation likelihood $R_{\pi}(a)$. Similarly, the generalized protocol presented immediately above is characterized as a POMDP with transition matrix $P$ and observation likelihood $\hat{R}_{\pi}(a)$. Notice that in our case the observation likelihoods are functions of $\pi$. We can then denote these POMDPs as $\theta = (P,R_{\pi})$ and $\hat{\theta} = (P,\hat{R}_{\pi})$, and their resultant optimal policies $\mu^*(\theta)$ and $\mu^*(\hat{\theta})$, respectively. Let $J_{\mu^*{\theta}}(\pi;\theta)$ and $J_{\mu^*(\theta)}(\pi;\hat{\theta})$ denote the discounted cumulative costs  incurred by these POMDPs when using policy $\mu^*(\theta)$. Similarly, $J_{\mu^*(\hat{\theta})}(\pi;\theta)$ and $J_{\mu^*(\hat{\theta})}(\pi;\hat{\theta})$ denote the discounted cumulative costs incurred by these POMDPs when using policy $\mu^*(\hat{\theta})$. These POMDP formulations have cost $C(e_i,u) = f\mathbb{I}_{\{e_i=2,u=1\}} + d\mathbb{I}_{\{e_1=1,u=2\}}$ and an implicit discount factor $\dcost = 1-p$ (see \cite{johnston2006opportunistic} for details). 

Now we can formulate a bound on the cumulative cost incurred when the quickest detector only has the estimate $\Lambda$ of psychological parameters:

\begin{theorem}
\label{thm:estimate}
Consider the quickest change detection protocols in which the quickest detector uses $\bar{T}(\pi,a)$ and $\hat{T}(\pi,a)$ for its public belief update. Denoting the corresponding POMDP characterizations by $\theta = (P,R_{\pi})$ and $\hat{\theta} = (P,\hat{R}_{\pi})$, respectively, and using the notation defined above, we have the inequality
\begin{align}
\label{ineq_result}
\begin{split}
    &J_{\mu^*(\hat{\theta})}(\pi,\theta) \leq J_{\mu^*(\theta)}(\pi,\theta) + 2K\|\theta - \hat{\theta}\| \\
    &K = \frac{1}{p}\max_{i,u}C(e_i,u)\\ 
    &\|\theta - \hat{\theta}\| = \sqrt{2}\sup_{\pi}{\max_i \sum_j P_{ij} [D(R_{j,\pi} \| \hat{R}_{j,\pi})]^{1/2}}
\end{split}
\end{align} 
where $D(R_{j,\pi} \| \bar{R}_{j,\pi}) = \sum_a R_{j,\pi}(a) \log(R_{j,\pi}(a) / \bar{R}_{j,\pi}(a))$ denotes the Kullback-Leibler divergence.
\end{theorem}
\begin{proof}
The proof slightly adapts that of Thorem 14.9.1 of \cite{krishnamurthy2016partially} and can be found in Appendix \ref{ap:pf1}.
\end{proof}

Observe that by Corollary \ref{cor1}, the KL Divergence term $D(R_{j,\pi} \| \bar{R}_{j,\pi})$ is continuous with respect to parameters $(\alpha,\lambda,\phi)$. Thus, the detection performance (given by cumulative cost $J_{\mu^*(\hat{\theta})}(\pi,\theta)$) of a quickest detector exploiting an \textit{estimate} of the human psychological parameters $(\alpha,\lambda,\phi)$ is bounded above by a continuous function of the inaccuracy (quantified by an appropriate norm in parameter space) of the parameter estimate. Informally, a change of $\epsilon$ in the parameter estimate will result in change in detection performance of $O(\epsilon)$. In this sense, the detection performance is robust to inaccuracy of the estimated human psychological parametrization.

\subsection{Blackwell Dominance Properties}

Here we present two theorems which will be used with our numerical study to reveal the existence of disjoint convex regions of the psychological parameter space which induce detection performance ordering. Theorem \ref{thm:val_great} states that if one steady-state Lindbladian distribution Blackwell dominates another, then the value function induced by the former is upper bounded by that of the latter. This allows us to reason about the detection performance (characterized by the value function) by investigating the property of Blackwell dominance between steady-state distributions. Theorem  \ref{thm:convex_dom} allows us to interpolate this performance ordering for all convex combinations of steady-state distributions which have this Blackwell dominance property. For ease of explanation, we say matrix $M_1$ is \textit{Blackwell dominating} with respect to matrix $M_2$ (and that $M_2$ is \textit{Blackwell dominated}) if there exists column stochastic matrix $B$ such that $M_2 = M_1B$.

\begin{theorem}
\label{thm:val_great}
Let $\Gamma^{\pi}$ and $\hat{\Gamma}^{\pi}$ be two steady-state action distributions, resulting from map \eqref{act_map} with prior $\pi$ and different Lindbladian parameterizations. Suppose there exists a stochastic (columns sum to 1) matrix $M^{\pi}$ such that $\Gamma^{\pi} = \hat{\Gamma}^{\pi}M^{\pi}$. Then, incorporating these distributions in the update \eqref{act_prob}, the value iteration algorithm \eqref{value_itr} yields $\mathcal{V}(\pi) \geq \hat{\mathcal{V}}(\pi)$, where $\mathcal{V}$ and $\hat{\mathcal{V}}$ are the value functions resulting from the use of distributions $\Gamma^{\pi}$ and $\hat{\Gamma}^{\pi}$, respectively.
\end{theorem}

\begin{proof}
See Appendix \ref{ap:val}.
\end{proof}

Somewhat more informally, this Theorem states that if steady-state distribution $\hat{\Gamma}^{\pi}$ \textit{Blackwell dominates} another steady state distribution $\Gamma^{\pi}$ for all $\pi \in \Pi$, then the quickest detector's performance (cumulative cost incurred) corresponding to the former distribution is better than that corresponding to the latter distribution.

\begin{theorem}
\label{thm:convex_dom}
Suppose there exist probability mass vectors $\hat{\Gamma},\Gamma_1,\Gamma_2 \in \reals^N (N \in \mathbb{N})$ and stochastic matrices $M_1, M_2 \in \reals^{N\times N}$ such that 
$\Gamma_1 = \hat{\Gamma}M_1$ and $\Gamma_2 = \hat{\Gamma}M_2$. Form $\Gamma_3 \in \reals^N$ as $\Gamma_3(a) = \gamma_a\Gamma_1(a) + (1-\gamma_a)\Gamma_2(a) ,\ \gamma_a \in [0,1], \  \forall a \in \{1,\dots,N\}$. Then there exists a stochastic matrix $M_3$ such that $\Gamma_3 = \hat{\Gamma}M_3$.
\end{theorem}
\begin{proof}
See Appendix \ref{ap:B}.
\end{proof} 

This Theorem states that if two steady-state distributions $\Gamma_1$ and $\Gamma_2$ are \textit{Blackwell dominated} by a third steady-state distribution $\hat{\Gamma}$, then any distribution $\Gamma_3$ which is a convex combination of $\Gamma_1$ and $\Gamma_2$ will also be Blackwell dominated by $\hat{\Gamma}$.

The following section presents our numerical results which reveal regions in the parameter space for which these conditions hold. Within these regions we can then be guaranteed performance ordering.

\section{Numerical Results}
Section~\ref{pris_dil} contains a numerical example (based on the Prisoner's Dilemma) which illustrates the ability of the Lindbladian decision model to account for violations of the sure-thing principle and demonstrates an implementation of quickest change detection for this context. Section~\ref{sec:numerical} provides several computational results which allow us to show the existence of disjoint convex parameter space regions for which the detection performance is strictly ordered (i.e. the quickest detector does better when the decision maker has psychological parameters in one region rather than the other). This numerical study also suggests that agent rationality plays a key role in this performance ordering, such that observing the decisions or more rational agents will increase detection performance. 

\subsection{Prisoner's Dilemma Numerical Example}
\label{pris_dil}
Here we provide a tutorial numerical example using  Prisoner's Dilemma problem \cite{martinez2016quantum}. We demonstrate the ability of the quantum decision theory to account for violations of the sure-thing principle, and provide an implementation of the quickest detector for a psychological parameterization which results in this violation. The key takeaways are that the Lindbladian model \eqref{Lindblad} can account for violations of the sure-thing principle (which cannot be accounted for by classical models), and that the quickest detector implementing this model still performs reasonably well while this violation is occurring. 

\subsubsection{Construction of Lindbladian Operator}
\label{lindblad_construct}
Here we illustrate the construction of the Lindbladian operator \eqref{Lindblad} for the Prisoner's Dilemma example. Suppose that the two underlying states of nature are whether or not the opponent defects $(D)$ or cooperates $(C)$, i.e. $\mathcal{X} = \{ 1 \textrm{(cooperate)}, 2 \textrm{(defect)}\}$. The actions of the agent are also to either cooperate or defect, i.e. $\mathcal{A} = \mathcal{X}$, and this action will depend on the agent's belief in the underlying state (the opponent's choice) and the payoff matrix. In this case we have the payoffs $a = u(C|C), b = u(C|D), c = u(D|D), d = u(D|C)$. We have a four-dimensional space of states $\mathcal{H} = \mathcal{H}_{\mathcal{X}} \otimes \mathcal{H}_{\mathcal{A}} = \{\ket{CC},\ket{DC},\ket{CD},\ket{DD}\}$ since two actions (cooperate or defect) are each associated to two states of nature (opponent cooperates or defects). To construct the Lindbladian operator \eqref{Lindblad}, we need to construct the Hamiltonian $H$ and the Cognitive Matrix $C(\lambda, \phi)$ \eqref{gamma_mn}. Following \eqref{Pi_mat}, we build the matrix $\Pi(\lambda)$ as 
\begin{equation}
\Pi(\lambda) = \left[
    \begin{matrix}  
    1-\mu(\lambda) & \mu(\lambda) & 0 & 0 \\  
    1-\mu(\lambda) & \mu(\lambda) & 0 & 0 \\
    0 & 0 & 1-\nu(\lambda) & \nu(\lambda) \\
    0 & 0 & 1-\nu(\lambda) & \nu(\lambda)
    \end{matrix}  \right]
\end{equation}

where $\mu(\lambda) = \frac{d^{\lambda}}{a^{\lambda} + d^{\lambda}}$ and $\nu(\lambda) = \frac{c^{\lambda}}{b^{\lambda} + c^{\lambda}}$. Suppose the agent has belief in the opponents action (underlying state) given by $\eta(x)$, such that $\eta(1) = P(\textrm{opponent cooperates})$ and $\eta(2) = P(\textrm{opponent defects})$. Then, following \eqref{B_Matrix} we have
\begin{equation}
    B = \left[
    \begin{matrix}  
    \eta(1) & 0 & \eta(2) & 0 \\  
    0 & \eta(1) & 0 & \eta(2) \\
    \eta(1) & 0 & \eta(2) & 0 \\
    0 & \eta(1) & 0 & \eta(2)
    \end{matrix}  \right] \ 
    H = \left[
    \begin{matrix}  
    1 & 1 & 0 & 0 \\  
    1 & 1 & 0 & 0 \\
    0 & 0 & 1 & 1 \\
    0 & 0 & 1 & 1
    \end{matrix}  \right]
\end{equation}

This simple Hamiltonian $H$ also agrees with \cite{martinez2016quantum}, as well as those used in quantum rankings of complex networks \cite{sanchez2012quantum}. 

\subsubsection{Violation of the Sure Thing Principle}
The Sure-Thing Principle, as described in Section~\ref{practicality} dates back to Savage \cite{savage1972foundations} and can be intuitively understood as follows. Suppose there exist two states of nature $A$ and $B$, and two actions $a_1$ and $a_2$. If $a_1$ is preferred to $a_2$ when the state is known to be $A$, and $a_1$ is also preferred to $a_2$ when the state is known to be $B$, then $a_1$ should be preferred to $a_2$ when the state is unknown or there is uncertainty in the state. However this principle was refuted in an experiment of Tversky and Shafir \cite{tversky1992disjunction} and this violation has been regularly experimentally reproduced since.  Note that this principle (see Section~\ref{practicality} for more formal definition) follows from the axioms of classical probability, namely the law of total probability. Thus, any classically probabilistic model for human decision making will be unable to account for such violations, hence the need for generalized quantum models. 

Busemeyer et. al. \cite{busemeyer2006quantum} investigate experimental violations of the Sure-Thing Principle (STP) in the context of the Prisoner's Dilemma, with payoff values $a$=20, $b$=5, $c$ = 10, and $d$ = 25. They find a defection rate of 91\% when the opponent is known to defect and 84\% when the opponent is known to cooperate. The STP is violated since the defection rate drops to 66\% when the choice of the opponent is unknown. We use the previous Lindbladian construction to reproduce this violation \cite{martinez2016quantum}, see Fig.~\ref{fig:STP_viol}. 

\begin{figure}[h!]
  \includegraphics[width=\linewidth]{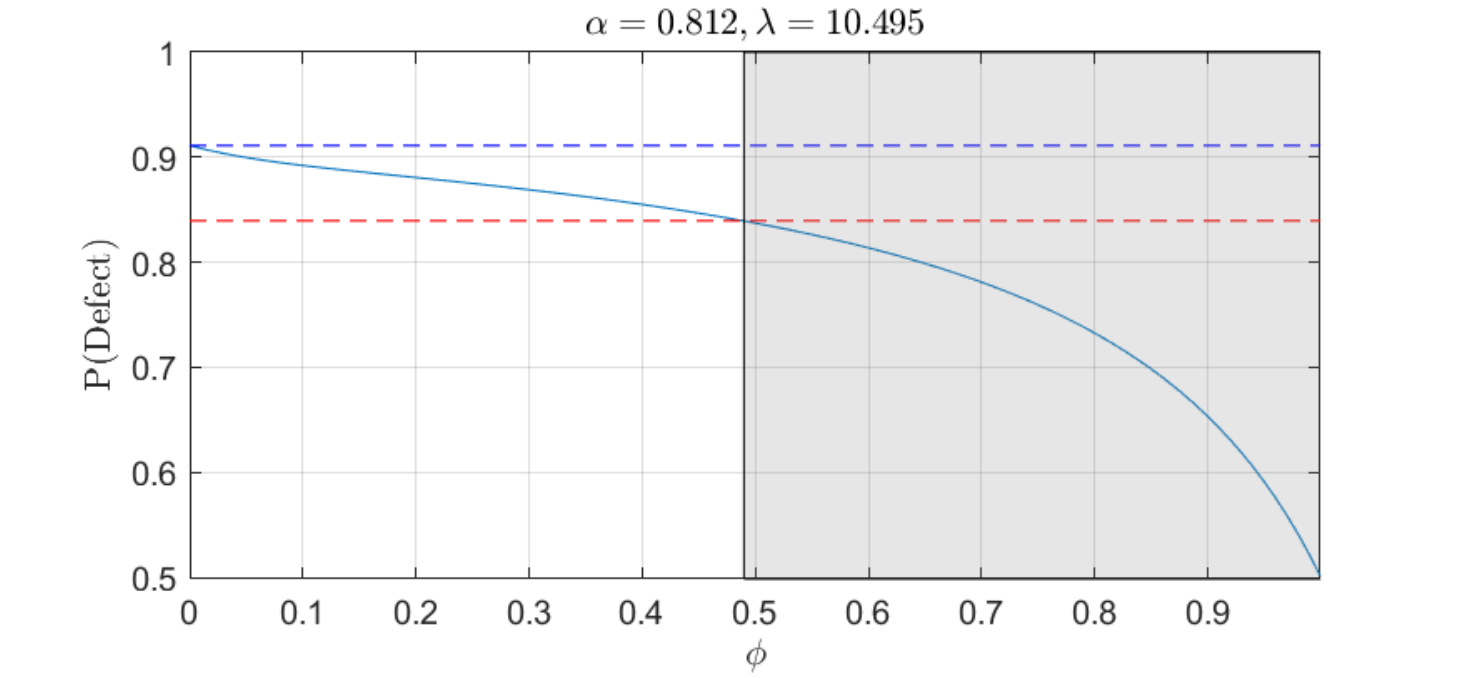}
  \caption{Violation of the Sure-Thing Principle. The probability of defection when the opponent is known to defect is 0.91 (dashed blue line). The probability of defection when the opponent is known to cooperate is 0.84 (dashed red line). The probability of defection when the choice of the opponent is unknown is given by the solid blue curve. Observe that within the shaded region, for $\phi > 0.49$, the probability of defection is not a \textit{convex combination} of those given certainty of the opponent's choice, and thus the STP is violated. This cannot be modeled with classical probabilistic decision models.}
  \label{fig:STP_viol}
\end{figure}

\subsubsection{Quickest Detector Implementation}
We implement the quickest change detection protocol of Section~\ref{QCD_prot} as well as the optimal policy computation of Section~\ref{QCD_OptPol}, within the context of this Prisoner's Dilemma problem. We take the underlying state to be the choice of the opponent, $x_n \in \{ 1 \textrm{(cooperate)}, 2 \textrm{(defect)}\}$. We assume this state jump changes from $2$ to $1$ according to a geometric distribution with parameter $p = 0.95$. We use the following simple measurement model: 
\begin{equation}
    \mathcal{Y} = \{1,2,3\}, \ \ B =  \left[\begin{matrix}  
        0.6 & 0.25 & 0.15 \\  
        0.15 & 0.25 & 0.6   
        \end{matrix} \right] 
\end{equation}   
where recall $\mathcal{Y}$ is the observation space, i.e. there are three possible observations. These observation are state-dependent with conditional probabilities given by $B_{x,y} = p(y_n = y  |x_n = x)$, for $y_n \in \{1,2,3\}, x_n \in \{1,2\}$. This observation is used in the computation \eqref{priv_bel} to obtain $\eta_n(x)$, which is input to the quantum decision map \eqref{act_map}.

The quantum decision maker chooses an action probabilistically according to the map \eqref{act_map}, with Lindbladian operator constructed as done in Section~\ref{lindblad_construct}. In order to incorporate a violation of the STP, we use the parametrization $(\alpha = 0.812, \lambda = 10.495, \phi = 0.9)$ (Observe from Fig.~\ref{fig:STP_viol} that this parametrization can result in an STP violation), and for simplicity we assume the quickest detector knows this parametrization. For false alarm penalty $f=5$ and delay penalty $d=1$, the quickest detector computes its optimal policy via the value iteration algorithm \eqref{value_itr}. This results in an optimal decision threshold of $\pi(1) = 0.834$ (denote this as $\pi'$), as illustrated in the optimal policy $\mu^*(\pi)$ plot in Fig.~\ref{fig:PD_OptPol}.

\begin{figure}[h!]
  \includegraphics[width=\linewidth]{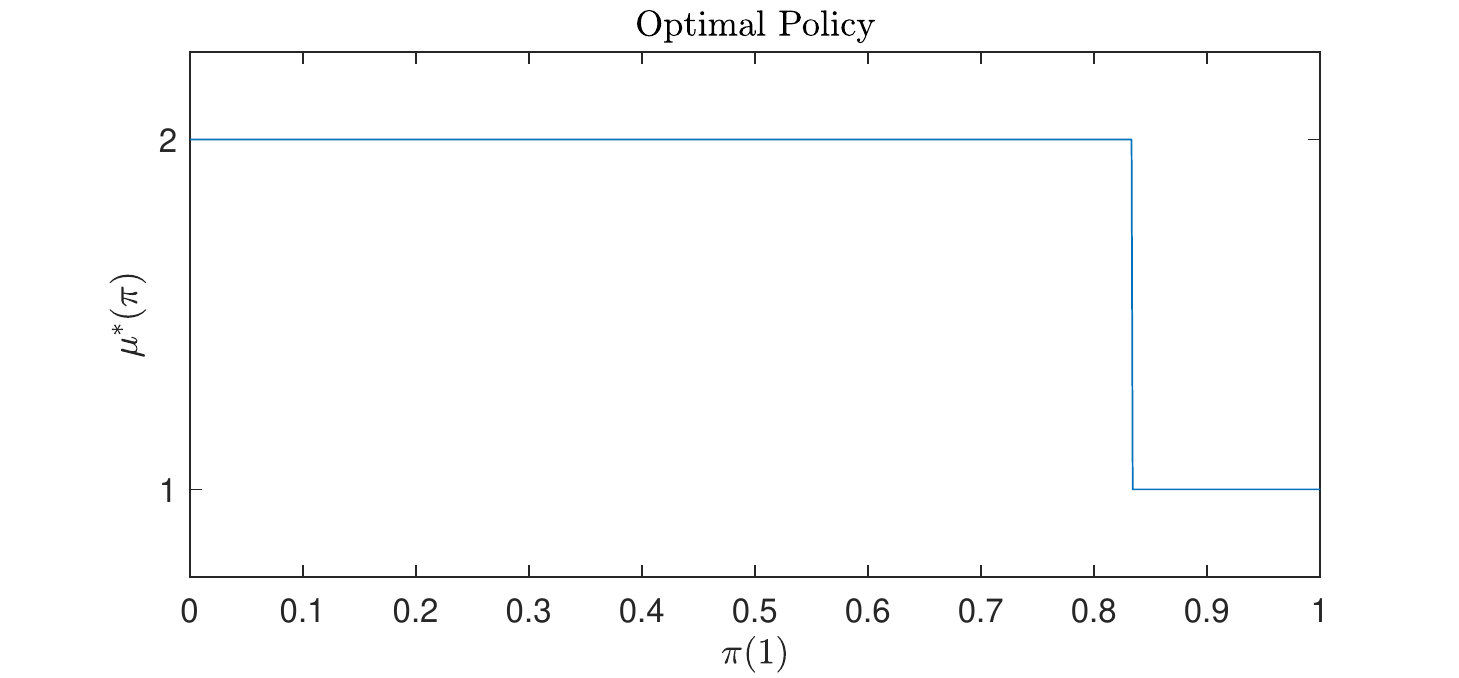}
  \caption{Optimal Policy of the quickest detector in the Prisoner's Dilemma setup In this case the optimal policy has a threshold at $\pi' = 0.834$. The sequential change detection protocol detailed in Section \ref{QCD_prot} results in an optimal policy $\mu^*(\pi)$ which has a single threshold. This is due to the concavity of the value function $\mathcal{V}(\pi)$, which is derived as a property resulting from the continuity of the distribution $\bar{\Gamma}_y^{\pi}(a)$ \ref{glob_dec_mak} with respect to $\pi$ and $y$.}
  \label{fig:PD_OptPol}
\end{figure}

Finally, we investigate the dependence of the optimal policy threshold on the false alarm to delay cost ratio. With delay cost fixed at one, we  computed the optimal policy threshold as a function of false alarm cost. This dependency is illustrated in Fig.~\ref{fig:OptThresVFal}. The figure  also shows this dependency for the \textit{classical} quickest detection protocol \eqref{clas_val_it}, where the detector directly observes the noisy sensor measurements. This reveals that the intermediate human decisions result in an optimal policy threshold which is larger than that resulting from classical quickest detection. 

\begin{figure}[h!]
  \includegraphics[width=\linewidth]{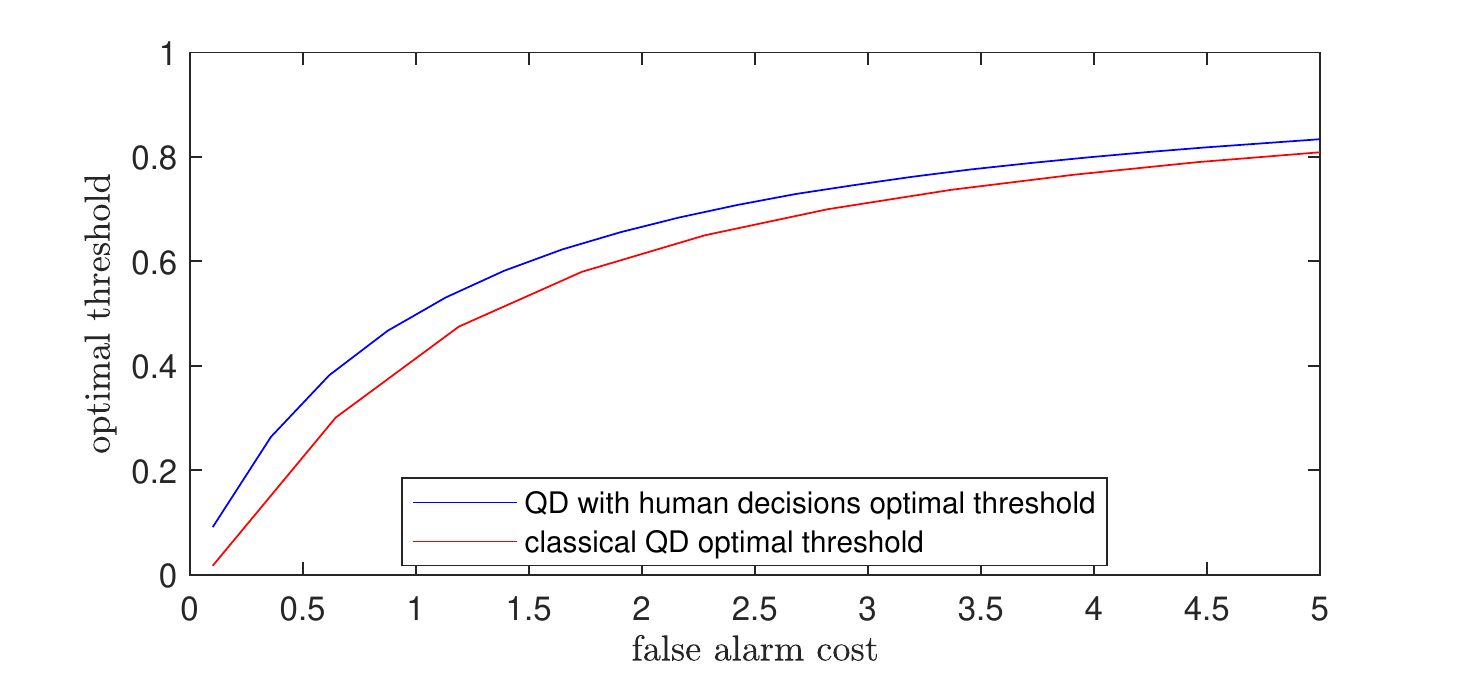}
  \caption{Optimal policy threshold for classical QCD (red), i.e detection directly using the noisy sensor measuremtns, and for our QCD protocl with intermediate human decisions (blue), plotted as functions of the false alarm cost (with delay cost = 1). Observe that the optimal threshold is always larger when there exist intermediate human decisions.}
  \label{fig:OptThresVFal}
\end{figure}

\subsection{Rationality Improves Detection Performance}
\label{sec:numerical}
We now present a series of computational results which allow us to verify the existence of convex regions in parameter space for which the performance of the quickest detector is strictly ordered. The first result, Example 1, verifies that for a specific subset of the parameter space, a convex combination of parameterizations which result in Blackwell dominated distributions also results in a Blackwell dominated distribution. The second result, Example~2,  verifies the converse, that a convex combination of parameterizations which result in distributions that are Blackwell dominating is also Blackwell dominating. These two in conjunction reveal that the Blackwell ordering is closed under convex parameter combinations, and thus we can interpolate this dominance ordering to hold between all points within the convex hulls of computed dominance points; this is studied in Example 3 below.

Denote $\bar{\Gamma}_{y,1}^{\pi}$ and $\bar{\Gamma}_{y,2}^{\pi}$ as the map \eqref{act_map} outputs for distinct Lindbladian parameterizations $(\alpha_1,\lambda_1,\phi_1)$ and $(\alpha_2,\lambda_2,\phi_2)$, respectively. Let $(\alpha_3,\lambda_3,\phi_3) = \epsilon(\alpha_1,\lambda_1,\phi_1) + (1-\epsilon)(\alpha_2,\lambda_2,\phi_2)$ for $\epsilon \in (0,1)$, and $\bar{\Gamma}_{y,3}^{\pi}$ be the resultant action distribution (from \eqref{act_map}) for Lindbladian parameterization $(\alpha_3,\lambda_3,\phi_3)$.

\textit{Numerical Verification 1}: We have verified numerically that for $\alpha_1,\alpha_2,\phi_1,\phi_2 \in [0.1,0.5]$, $\lambda_1,\lambda_2 \in [10,100]$, and $\epsilon \in (0,1)$, there exists $\{\gamma_a\}, a \in \{1,\dots,M\}$ such that $\bar{\Gamma}_{y,3}^{\pi}(a) = \gamma_a\bar{\Gamma}_{y,1}^{\pi} + (1-\gamma_a)\bar{\Gamma}_{y,2}^{\pi}$. In words, within the parameter confines defined, interpolating between two Lindbladian parameterizations via a convex combination results in an action distribution for which each action probability is a convex combination of (lies between) the action probabilities resulting from the initial two parameterizations. 

This numerical verification, along with the following two Theorems, will be used to prove our first computational result which reveals that, for certain regions in the parameter space, a Blackwell dominance order is closed under convex combinations.

\textit{Result 1. Performance dominance is closed under convex combinations of dominated distributions:} \\
Consider Lindbladian parameterizations $(\alpha_1,\lambda_1,\phi_1) \in [0.1,0.5]\times[10,100]\times[0.1,0.5], (\alpha_2,\lambda_2,\phi_2) \in [0.1,0.5]\times[10,100]\times[0.1,0.5], (\alpha_3,\lambda_3,\phi_3) \in [0.5,1]\times[10,100]\times[0,1]$  and respective resultant action distributions $\Gamma^{\pi}_1,\Gamma^{\pi}_2,\Gamma^{\pi}$ from the map \eqref{act_map} with prior $\pi$. Let $(\bar{\alpha},\bar{\lambda},\bar{\phi}) = \epsilon(\alpha_1,\lambda_1,\phi_1) + (1-\epsilon)(\alpha_2,\lambda_2,\phi_2)$ for $\epsilon \in [0,1]$, and $\Gamma_3^{\pi}$ the resultant action distribution. Denote $\mathcal{V}(\pi),\mathcal{V}_1(\pi),\mathcal{V}_2(\pi),\mathcal{V}_3(\pi)$ the value functions resulting from the value iteration algorithm \eqref{value_itr} using $\Gamma^{\pi},\Gamma^{\pi}_1,\Gamma^{\pi}_2,\Gamma_3^{\pi}$ in \eqref{glob_dec_mak}, respectively. Suppose there exists stochastic matrices $M_1^{\pi},M_2^{\pi}$ such that $\Gamma_1^{\pi} = \Gamma^{\pi}M_1^{\pi}$ and $\Gamma^{\pi}_2 = \Gamma^{\pi} M_2^{\pi}$. Then we have $\mathcal{V}_1(\pi) \geq \mathcal{V}(\pi), \mathcal{V}_2(\pi) \geq \mathcal{V}(\pi)$ and $\mathcal{V}_3(\pi) \geq \mathcal{V}(\pi)$ $\ \forall \pi \in \Pi.$\\
\textit{Verification}:
By the numerical verification 1, we have that 
\[ \exists \{\gamma_a\}_{a=1}^{\NumAct} : \ \Gamma^{\pi}(a) = \gamma_a\Gamma^{\pi}_1(a) + (1-\gamma_a)\Gamma^{\pi}_2(a)\]
Thus, by Theorem \ref{thm:convex_dom}, there exists a stochastic matrix $M_3^{\pi}$ such that $\Gamma_3^{\pi} = \Gamma^{\pi}M_3$. Then by invoking Theorem \ref{thm:val_great} using each equality $\Gamma_1^{\pi} = \Gamma^{\pi}M_1^{\pi},\  \Gamma_2^{\pi} = \Gamma^{\pi}M_2^{\pi}, \ \Gamma_3^{\pi} = \Gamma^{\pi}M_3^{\pi}$, the results follow.

Now we provide another numerical verification and two computational results which allow us to conclude closure of Blackwell dominance orderings in the opposite direction.

Denote $\bar{\Gamma}_{y,1}^{\pi}$ and $\bar{\Gamma}_{y,2}^{\pi}$ as the map \eqref{act_map} outputs for distinct Lindbladian parameterizations $(\alpha_1,\lambda_1,\phi_1)$ and $(\alpha_2,\lambda_2,\phi_2)$, respectively. Let $(\alpha_3,\lambda_3,\phi_3) = \epsilon(\alpha_1,\lambda_1,\phi_1) + (1-\epsilon)(\alpha_2,\lambda_2,\phi_2)$ for $\epsilon \in (0,1)$, and $\bar{\Gamma}_{y,3}^{\pi}$ be the resultant action distribution (from \eqref{act_map}) for Lindbladian parameterization $(\alpha_3,\lambda_3,\phi_3)$.

\textit{Numerical Verification 2}: We have verified numerically that for $\alpha_1,\alpha_2 \in [0.8,1] ,\phi_1,\phi_2 \in [0.1,0.5]$, $\lambda_1,\lambda_2 \in [10,100]$, and $\epsilon \in (0,1)$, there exists $\{\gamma_a\}, a \in \{1,\dots,M\}$ such that $\bar{\Gamma}_{y,3}^{\pi}(a) = \gamma_a\bar{\Gamma}_{y,1}^{\pi} + (1-\gamma_a)\bar{\Gamma}_{y,2}^{\pi}$. In words, within the parameter confines defined, interpolating between two Lindbladian parameterizations via a convex combination results in an action distribution for which each action probability is a convex combination of (lies between) the action probabilities resulting from the initial two parameterizations. \\

\textit{Result 2.  Performance dominance is closed under convex combinations of dominating distributions:} \\
Consider Lindbladian parameterizations $(\alpha_1,\lambda_1,\phi_1) \in [0.8,1]\times[10,100]\times[0.1,0.5], (\alpha_2,\lambda_2,\phi_2) \in [0.8,1]\times[10,100]\times[0.1,0.5], (\alpha_3,\lambda_3,\phi_3) \in [0.5,1]\times[10,100]\times[0,1]$  and respective resultant action distributions $\Gamma^{\pi}_1,\Gamma^{\pi}_2,\Gamma^{\pi}$ from the map \eqref{act_map} with prior $\pi$. Let $(\bar{\alpha},\bar{\lambda},\bar{\phi}) = \epsilon(\alpha_1,\lambda_1,\phi_1) + (1-\epsilon)(\alpha_2,\lambda_2,\phi_2)$ for $\epsilon \in [0,1]$, and $\Gamma_3^{\pi}$ the resultant action distribution. Denote $\mathcal{V}(\pi),\mathcal{V}_1(\pi),\mathcal{V}_2(\pi),\mathcal{V}_3(\pi)$ the value functions resulting from the value iteration algorithm \eqref{value_itr} using $\Gamma^{\pi},\Gamma^{\pi}_1,\Gamma^{\pi}_2,\Gamma_3^{\pi}$ in \eqref{glob_dec_mak}, respectively. Suppose there exists stochastic matrices $M_1^{\pi},M_2^{\pi}$ such that $\Gamma^{\pi} = \Gamma_1^{\pi}M_1^{\pi}$ and $\Gamma^{\pi} = \Gamma_2^{\pi} M_2^{\pi}$. Then we have $\mathcal{V}(\pi) \geq \mathcal{V}_1(\pi), \mathcal{V}(\pi) \geq \mathcal{V}_2(\pi)$ and $\mathcal{V}(\pi) \geq \mathcal{V}_3(\pi)$ $\ \forall \pi \in \Pi.$

\textit{Verification}:
By the numerical verification 2, we have that 
\[ \exists \{\gamma_a\}_{a=1}^{\NumAct} : \ \Gamma^{\pi}(a) = \gamma_a\Gamma^{\pi}_1(a) + (1-\gamma_a)\Gamma^{\pi}_2(a)\]
Thus, by Result 3, there exists a stochastic matrix $M_3^{\pi}$ such that $\Gamma^{\pi} = \Gamma_3^{\pi}M_3$. Then by invoking Theorem \ref{thm:val_great} using each equality $\Gamma^{\pi} = \Gamma_1^{\pi}M_1^{\pi},\  \Gamma^{\pi} = \Gamma_3^{\pi}M_2^{\pi}, \ \Gamma^{\pi} = \Gamma_3^{\pi}M_3^{\pi}$, the results follow.
\\

\textit{Result 3. Blackwell dominance is closed under convex combinations of dominating distributions:} \\
\label{thm:convex_dom2}
Take probability mass vectors $\hat{\Gamma},\Gamma_1,\Gamma_2 \in \reals^N (N \in \mathbb{N})$. Suppose there exist invertible column stochastic matrices $M_1, M_2 \in \reals^{N\times N}$ with the property that $M_1^{-1}, M_2^{-1}$ are "loosely column stochastic", meaning that each column sums to 1 but need not necessarily have all non-negative elements, such that 
$\hat{\Gamma} = \Gamma_1 M_1$ and $\hat{\Gamma} = \Gamma_2 M_2$. Form $\Gamma_3 \in \reals^N$ as $\Gamma_3(a) = \gamma_a\Gamma_1(a) + (1-\gamma_a)\Gamma_2(a) ,\ \gamma_a \in [0,1], \  \forall a \in \{1,\dots,N\}$. Then there exists a stochastic matrix $M_3$ such that $\hat{\Gamma} = \Gamma_3 M_3$.

\textit{Verification}:
We have $\hat{\Gamma} = \Gamma_1 M_1$ and $\hat{\Gamma} = \Gamma_2 M_2$, and thus 
\begin{align}
    \begin{split}
    \label{inverse_cond}
        &\hat{\Gamma}M_1^{-1} = \Gamma_1, \ \  \hat{\Gamma}M_2^{-1} = \Gamma_2, \\ &\sum_{i=1}^{\NumAct} M_1^{-1}(a,i)  =  \sum_{i=1}^{\NumAct} M_2^{-1}(a,i) = 1 \ \forall a \in \{1,\dots,\NumAct\}
    \end{split}
\end{align}
Observe that Theorem \ref{thm:convex_dom} also holds if $M_1$ and $M_2$ are only assumed to by loosely column stochastic, in which case $M_3$ is also loosely column stochastic. Thus Theorem \ref{thm:convex_dom} can be invoked on the equalities \eqref{inverse_cond} to construct loosely column stochastic matrix $M_3^{-1}$ such that $\hat{\Gamma}M_3^{-1} = \Gamma_3$. We observe numerically that any such construction $M_3^{-1}$, intialized by matrices $M_1, M_2$ constructed such that $\hat{\Gamma} = \Gamma_1M_1, \hat{\Gamma} = \Gamma_2M_2$ for $\hat{\Gamma},\Gamma_1,\Gamma_2$ steady state Lindbladian distributions and scalars $\{\gamma_a\}_{a=1}^{\NumAct} \in [0,1]^{\NumAct}$, has inverse $M_3 = (M_3^{-1})^{-1}$ which is column stochastic.

Here we demonstrate a numerical consequence of the preceding Theorems. There exist disjoint convex regions $R_1$ and $R_2$ in the parameter space $\{\alpha,\lambda,\phi\}$, such that the value functions $\mathcal{V}_{p_1}$ and $\mathcal{V}_{p_2}$ resulting from  Lindbladian parameterizations \eqref{act_map} $p_2 = \{\alpha_1,\lambda_1,\phi_1\} \in R_1$, $p_2= \{\alpha_2,\lambda_2,\phi_2\} \in R_2$ satisfy
\[\mathcal{V}_{p_2}(\pi) \geq \mathcal{V}_{p_1}(\pi) \ \forall \pi \in \Pi, \ \forall p_1 \in R_1, p_2 \in R_2 \]

\begin{figure}[h!]
  \includegraphics[width=\linewidth]{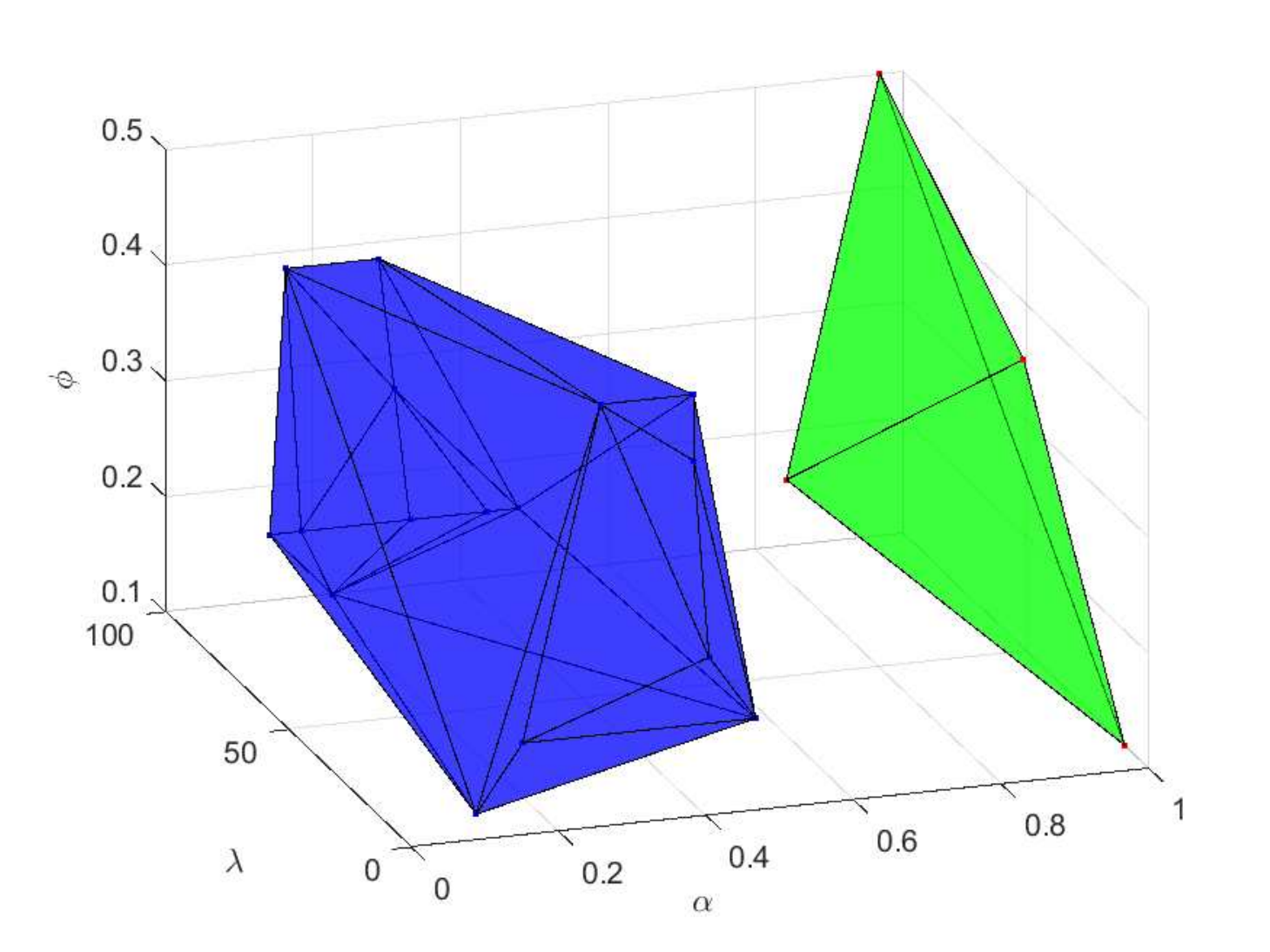}
  \caption{Convex parameter regions inducing a Blackwell Order. Denote the green and blue regions by $R_1$ and $R_2$, respectively. For any points $p_1 \in R_1, p_2 \in R_2$, the parameterized Lindbladian models $\mathcal{L}_{p_1}, \mathcal{L}_{p_2}$ which characterize the map \eqref{act_map}, when used in the local decision making process of Fig. \ref{fig:QCD_Proc}, result in respective optimal value functions $\mathcal{V}_{p_1}, \mathcal{V}_{p_2}$ satisfying $\mathcal{V}_{p_2}(\pi) \geq \mathcal{V}_{p_1}(\pi) \ \forall \pi \in \Pi$. The takeaway message is that quickest change detection performance is strictly better when the human decision maker has psychological parameters in the green region compared to the blue region.}
  \label{fig:ConvDom}
\end{figure}

Figure \ref{fig:ConvDom} demonstrates two such regions, the green region corresponding to $R_1$ and the blue region corresponding to $R_2$. In words, this result means that for the quickest change detection system of Fig. \ref{fig:QCD_Proc}, given local human decision makers acting with a decision making process characterized by the Lindbladian evolution with parameter $p_1 \in R_1$ and given another acting with process parameterized by $p_2 \in R_2$, the quickest detector's optimal cost in the former case is upper bounded by that of the latter case. Intuitively, this means that the quickest change detection system performs strictly better when the local decision maker has psychological parametrization in $R_1$ than when the local decision maker has psychological parametrization in $R_2$. 

Examining the regions $R_1$ and $R_2$ in  Fig. \ref{fig:ConvDom}, it can be inferred that the value of the $\alpha$ parameter plays a significant role in the system performance. Recall that $\alpha$  interpolates between the purely 'rational' Markovian dynamics and the 'irrational' quantum dynamics. This suggests that higher rationality in the local decision maker will translate to improved  detection performance on the part of the quickest detector. 

\section{Conclusion}
At an abstract level, statistical signal processing deals with signal estimation using sensors, while psychology aims to model (understand) human decision making.
We have presented a sequential  change detection framework involving a human decision maker (modeled via quantum decision theory from recent results in psychology), sensor and a quickest detector.
Quantum decision theory captures important features of human decision making such as order effects and violation of the sure thing principle (total probability rule).
 The framework of this paper contributes to the area of human-sensor interface design and analysis. 
 
 The aim of our quickest detection formulation was to detect a change in the underlying state by observing the  human decisions which are influenced by the state.
 We characterized the structure of the quickest detection policy. We  showed that the optimal policy has a single threshold, and that the optimal cost incurred is lower bounded by that of the classical quickest detection framework, suggesting that the intermediate human decisions cannot improve the detection performance. We have also provided an upper bound on the cumulative cost incurred by the analyst when only a probabilistic estimate of the human's psychological parametrization is available. This upper bound is given in terms of the cumulative cost incurred when the analyst has perfect knowledge of the parameters and the KL-Divergence between action distributions induced by the true and estimated parametrizations. Finally, we also showed that certain humans in the loop are better than others (w.r.t. quantum decision parameters) in terms of the overall cost in quickest detection. 

One aspect we have not considered in this paper is performance analysis of the quickest detector with a quantum decision maker. For performance analysis of quickest detection we refer to the important papers \cite{PH08}, \cite{https://doi.org/10.48550/arxiv.2110.01581}, \cite{TM10}, \cite{Mou86}, and references therein. 
It is worthwhile investigating  how the Linbladian parameters for the human decision maker affect the performance of the quickest detector.

\appendices
\section{Proofs of Theorems}
\subsection{Proof of Theorem \ref{thm:sing_thres}}
\label{pf:sing_thres}

\begin{proof}
We first show that the action distribution $\bar{\Gamma}_y^{\pi}(a) := \PR(a | \pi,y)$ induced as the unique steady-state distribution for the Lindbladian $\CLind$ with parameters $(\alpha,\lambda,\phi)$ and initial state $\rho_0$, via the map \eqref{act_map}, is a continuous function of $\pi$ and $y$. We then use this within an induction argument in the value iteration algorithm \eqref{value_itr} to complete the argument. 

The vectorized solution of \eqref{Lind_ev} \cite{martinez2016quantum}, for any vectorized initial condition $\vec{\rho}(0)$, is \[\vec{\rho}(t) = \exp(\CLind t)\vec{\rho}(0)\]
 Fix $t > 0$ and, examining the structure of the operator $\CLind$ \eqref{Lindblad}, consider the map \[\Lambda_{(m,n)}^t: \eta_k(\cdot) \to [\exp(\CLind t) \vec{\rho}(0)]_{(m,n)}\] By inspection, each element $[\CLind]_{(m,n)}$ of \eqref{Lindblad} is continuous with respect to $\gamma_{(m,n)}$ and thus with respect to $\eta_k(\cdot)$ (by \eqref{gamma_mn} and \eqref{B_Matrix}). Thus the map $\Lambda_{(m,n)}^t$ is continuous with respect to $y$ and $\pi \ \forall (m,n) \in [1,d]^2, t \in \reals_+$. Also observe that $\eta_k(\cdot)$ is a continuous function of $\pi$ (for a fixed observation $y$) and $y$ (for a fixed prior $\pi$), as a Bayesian update. So we have that the action distribution at time $t$
 \[\bar{\Gamma}_y^{\pi}(a,t) = \Tr(P_a \exp(\CLind t) \vec{\rho}(0) P_a^{\dagger})\]
 is continuous with respect to $y$ and $\pi$ $\ \forall  (m,n) \in [1,d]^2, t \in \reals_+$. Thus the stationary distribution 
 \[\bar{\Gamma}_y^{\pi}(a) = \lim_{t \to \infty} \bar{\Gamma}_y^{\pi}(a,t)\] 
 is continuous with respect to $\pi$ and $y$.
We now use induction on the value iteration algorithm \eqref{value_itr}. The algorithm begins with \[\mathcal{V}_0(\pi) = 0 \ \forall \pi \in \Pi\] 
Thus $\mathcal{V}_0(\pi)$ is trivially concave. Now assume $\mathcal{V}_k(\pi)$ is concave for some $k \in \mathbb{N}$. We have the update 
\begin{align}
    \begin{split}
        &\mathcal{V}_{k+1}(\pi) = \min\{C(\pi,1), \\ &\ C(\pi,2) +  \sum_{a \in \mathcal{A}_1 \times \mathcal{A}_2} \mathcal{V}_k (\bar{T}(\pi,a))\bar{\sigma}(\pi,a)\}
    \end{split}
\end{align}
Observe from \eqref{value_itr} that $\mathcal{V}_k(\pi)$ is positively homogeneous; that is, for any $\alpha>0$, $\mathcal{V}_k(\alpha \pi) = \alpha\mathcal{V}_k(\pi)$. Choosing $\alpha = \bar{\sigma}(\pi,a)$ yields
\begin{align}
    \begin{split}
        &\mathcal{V}_{k+1}(\pi) = \min\{C(\pi,1), \\ &\ C(\pi,2) +  \sum_{a \in \mathcal{A}_1 \times \mathcal{A}_2} \mathcal{V}_k (R_{\pi}(a)P'\pi)\}
    \end{split}
\end{align}
Recall that $R_{\pi}(a)$ is computed via \eqref{act_prob}, and thus we have that $R_{\pi}(a)$ is a continuous function of $\pi, \ \forall a \in \mathcal{A}$. Also recall that $C(\pi,1)$ and $C(\pi,2)$ are linear in $\pi$. Thus $\mathcal{V}_{k+1}(\pi)$ is concave. This completes the induction step. Now the value iteration algorithm \eqref{value_itr} converges pointwise, so the optimal value function 
\[\mathcal{V}(\pi) = J_{\mu}^*(\pi) = \lim_{k \to \infty} \mathcal{V}_k(\pi)\]
is concave. This immediately implies, by \eqref{bellman}, that the optimal policy $\mu^*(\pi)$ cannot have more than one threshold.
\end{proof}

\begin{corollary}
\label{cor1}
Recall that $\bar{\Gamma}_y^{\pi}(a)$ is inherently dependent on a choice of parameters $(\alpha,\lambda,\phi)$. From the Proof of Theorem \ref{thm:sing_thres} and by inspection of the Lindbladian structure \eqref{Lindblad}, we have that $\bar{\Gamma}_y^{\pi}(a)$ is continuous with respect to parameters $(\alpha,\lambda,\phi)$. Then also $R_{x,\pi}(a)$ \eqref{act_prob} is continuous with respect to $(\alpha,\lambda,\phi)$.
\end{corollary}

\subsection{Proof of Theorem \ref{thm:lower_bound}}
\label{pf:lower_bound}
\begin{proof}
It is well known \cite{krishnamurthy2016partially} that the value function $\underline{\mathcal{V}}_k(\pi)$ for classical quickest detection is concave over $\Pi$ for any $k$. First recall that the quantum decision maker's action probability distribution, from the map \eqref{act_map}, given public belief $\pi$ and observation $y \in \mathcal{Y}$ is denoted 
\[\bar{\Gamma}_y^{\pi}(a) := \PR(a | \pi,y) \]
as in \eqref{act_prob}. First we express the public belief update \eqref{glob_dec_mak} in terms of the private belief update \eqref{priv_bel} as 
\begin{align}
    \begin{split}
        &\bar{T}(\pi,a) = \sum_{y \in \mathcal{Y}}T(\pi,y)\frac{\sigma(\pi,y)}{\bar{\sigma}(\pi,a)}\bar{\Gamma}_y^{\pi}(a), \\ &\bar{\sigma}(\pi,a) = \sum_{y \in \mathcal{Y}}\sigma(\pi,y)\bar{\Gamma}_y^{\pi}(a)
    \end{split}
\end{align}
Since $\underline{V}_k(\cdot)$ is concave for $\pi \in \Pi$, using Jensen's Inequality it follows that 
\begin{align*}
    \begin{split}
        &\underline{\mathcal{V}}_k(\bar{T}(\pi,a)) = \underline{\mathcal{V}}_k \left(\sum_{y \in \mathcal{Y}}T(\pi,y) \frac{\sigma(\pi,y)}{\bar{\sigma}(\pi,a)}\bar{\Gamma}_y^{\pi}(a) \right) \\ &\geq \sum_{y \in \mathcal{Y}}\underline{\mathcal{V}}_k(T(\pi,y))\frac{\sigma(\pi,y)}{\bar{\sigma}(\pi,a)}\bar{\Gamma}_y^{\pi}(a)
    \end{split}
\end{align*}
Therefore 
\begin{equation}
\label{ineq}
    \sum_{a \in \mathcal{A}}\underline{\mathcal{V}}_k(\bar{T}(\pi,a))\bar{\sigma}(\pi,a) \geq \sum_{y \in \mathcal{Y}} \underline{\mathcal{V}}_k(T(\pi,y))\sigma(\pi,y)
\end{equation}
Now we proceed by induction on the value iteration algorithm \eqref{value_itr}. Assume $\mathcal{V}_k(\pi) \geq \underline{\mathcal{V}}_k(\pi)$ for $\pi \in \Pi$. Then
\begin{align*}
    \begin{split}
        &C(\pi,2) + \sum_{a \in \CA}\mathcal{V}_k(\bar{T}(\pi,a))\bar{\sigma}(\pi,a) \\ &\geq C(\pi,2) + \sum_{a \in \CA}\underline{\mathcal{V}}_k(\bar{T}(\pi,a))\bar{\sigma}(\pi,a) \\ 
        &\geq C(\pi,2) + \sum_{y \in \mathcal{Y}}\underline{\mathcal{V}}_k(T(\pi,y))\sigma(\pi,y)
    \end{split}
\end{align*}
where the second inequality follows from \eqref{ineq}. Thus $\mathcal{V}_{k+1}(\pi) \geq \underline{\mathcal{V}}_{k+1}(\pi)$ and the induction step is complete. The value iteration algorithm converges pointwise and so $\mathcal{V}(\pi) \geq \underline{\mathcal{V}}(\pi)$, completing the proof.
\end{proof}

\subsection{Proof of Theorem \ref{thm:estimate}}
\label{ap:pf1}

\begin{proof}
The proof is adapted from that provided for Theorem 14.9.1 of \cite{krishnamurthy2016partially}. We first note that by the reasoning of \cite{johnston2006opportunistic} (Appendix, Proof of Theorem 2), POMDPs $\theta $ and $\hat{\theta}$ have implicit discount factor \[\gamma = 1-p \] The cumulative cost incurred by applying policy $\mu(\pi)$ to model $\theta$ satisfies at time $n$
\begin{equation}
    J_{\mu}^{(n)}(\pi;\theta) = C'_{\mu(\pi)}\pi + \dcost \sum_{a}J_{\mu}^{(n-1)}(T(\pi,a,\mu(\pi);\theta)\sigma(\pi,a,\mu(\pi);\theta).
\end{equation}
Therefore, the absolute difference in cumulative cost for models $\theta,\bar{\theta}$ satisfies 
\begin{align}
\label{cum_cost_diff}
    \begin{split}
        &|J_{\mu}^{(n)}(\pi;\theta) - J_{\mu}^{(n)}(\pi;\hat{\theta})| \\
        &\leq \dcost \sum_a \sigma(\pi,a,\mu(\pi);\theta)|J_{\mu}^{(n-1)}(T(\pi,a,\mu(\pi);\theta)) - \\&J_{\mu}^{(n-1)}(T(\pi,a,\mu(\pi);\hat{\theta}))| \\
        &+ \dcost \sum_a J_{\mu}^{(n-1)}(T(\pi,a,\mu(\pi);\hat{\theta}))|\sigma(\pi,a,\mu(\pi);\theta)\\&-\sigma(\pi,a,\mu(\pi);\hat{\theta})| \\
        &\leq \dcost \sup_{\pi \in \Pi}|J_{\mu}^{(n-1)}(\pi,\theta) - J_{\mu}^{(n-1)}(\pi;\hat{\theta})|\sum_a \sigma(\pi,a,\mu(\pi);\theta) \\
        & + \dcost \sup_{\pi \in \Pi}J_{\mu}^{(n-1)}(\pi;\hat{\theta})\sum_a |\sigma(\pi,a,\mu(\pi);\theta - \sigma(\pi,a,\mu(\pi);\hat{\theta})|.
    \end{split}
\end{align}
Observe that $\sum_a \sigma(\pi,a.\mu(\pi);\theta) = 1$. Then evaluating $\sigma(\pi,a;\theta) = \textbf{1}'R_{\pi}(a)P'\pi$ and $\sigma(\pi,a;\hat{\theta}) = \textbf{1}'\hat{R}_{\pi}(a)P'\pi$ yields

\begin{align}
    \begin{split}
        &\sum_a | \sigma(\pi,a,\mu(\pi);\theta) - \sigma(\pi,a.\mu(\pi);\hat{\theta})| \\
        &\leq \sum_a \sum_i \sum_j |R_{j,\pi}(a)P_{ij} - \hat{R}_{j,\pi}(a)P_{ij}|\pi(i) \\
        &\leq \max_i \sum_a \sum_j |R_{j,\pi}(a)P_{ij} - \hat{R}_{j,\pi}(a)P_{ij} | \\
        & = \max_i \sum_j P_{ij} \sum_a |R_{j,\pi}(a) - \hat{R}_{j,\pi}(a) | \\
        &\leq \sqrt{2} \max_i \sum_j P_{ij} [D(R_{j,\pi}\| \hat{R}_{j,\pi})]^{1/2}
    \end{split}
\end{align}
where the last inequality follows from Pinsker's inequality. Now we also have
\begin{equation}
    \sup_{\pi \in \Pi}J_{\mu}^{(n-1)}(\pi;\hat{\theta}) \leq \frac{1}{1-\dcost}\max_i C(e_i,u).
\end{equation}
We use these bounds in \eqref{cum_cost_diff} to obtain
\begin{align}
\label{14.45}
    \begin{split}
       & \sup_{\pi \in \Pi}|J_{\mu}^{(n)}(\pi;\theta) - J_{\mu}^{(n)}(\pi;\hat{\theta})| \\
       &\leq \dcost \sup_{\pi \in \Pi} |J_{\mu}^{(n-1)}(\pi;\theta) - J_{\mu}^{(n-1)}(\pi;\hat{\theta})| \\
       &+ \frac{\sqrt{2}\dcost}{1-\dcost} \max_i C(e_i,u) \sup_{\pi \in \Pi}\sum_j P_{ij} [D(R_{j,\pi}\| \hat{R}_{j,\pi})]^{1/2}
    \end{split}
\end{align}
Now starting with $J_{\mu}^{(0)}(\pi;\theta) = J_{\mu}^{(0)}(\pi;\hat{\theta}) = 0$, unraveling \eqref{14.45} yields
\begin{align}
\label{14.42}
\begin{split}
    &\sup_{\pi \in \Pi}|J_{\mu}(\pi;\theta) - J_{\mu}(\theta;\hat{\theta})| \\
    &\leq\frac{\sqrt{2}}{1-\dcost} \max_i C(e_i,u) \sup_{\pi \in \Pi}\sum_j P_{ij} [D(R_{j,\pi}\| \hat{R}_{j,\pi})]^{1/2}
\end{split}
\end{align}

Note that trivially we have
\begin{align}
    \begin{split}
     &J_{\mu^*(\hat{\theta})}(\pi,\theta) \\&\leq J_{\mu^*(\hat{\theta})}(\pi,\hat{\theta}) + \sup_{\pi}|J_{\mu^*(\hat{\theta}}(\pi,\theta) - J_{\mu^*(\hat{\theta}}(\pi,\hat{\theta})| \\
     & J_{\mu^*(\theta)}(\pi,\hat{\theta}) \\&\leq J_{\mu^*(\theta)}(\pi,\theta) + \sup_{\pi}|J_{\mu^*(\theta}(\pi,\theta) - J_{\mu^*(\theta}(\pi,\hat{\theta})|
    \end{split}
\end{align}
Also, by definition, $J_{\mu^*(\hat{\theta})}(\pi,\hat{\theta}) \leq J_{\mu^*(\theta)}(\pi,\hat{\theta})$, so
\begin{align}
    \begin{split}
        &J_{\mu^*(\hat{\theta})}(\pi,\theta) \leq J_{\mu^*(\theta)}(\pi,\theta) \\
        &+ \sup_{\pi}|J_{\mu^*(\hat{\theta})}(\pi,\theta) - J_{\mu^*(\hat{\theta})}(\pi,\hat{\theta})| \\ 
        &+ \sup_{\pi}|J_{\mu^*(\theta)}(\pi,\theta) - J_{\mu^*(\theta)}(\pi,\hat{\theta})| \\
        & \leq J_{\mu^*(\theta)}(\pi,\theta) + 2\sup_{\mu} \sup_{\pi}|J_{\mu^*(\hat{\theta})}(\pi,\theta) - J_{\mu^*(\hat{\theta})}(\pi,\hat{\theta})|.
    \end{split}
\end{align}
Then from \eqref{14.42}, \eqref{ineq_result} follows.
\end{proof}

\subsection{Proof of Theorem \ref{thm:val_great}}
\label{ap:val}
\begin{proof}
Consider the update \eqref{glob_dec_mak} and define  
   \[R_{x,\pi}(a) = \int_{\mathcal{Y}}{\Gamma}_{y}^{\pi_{n-1}}(a)B_{x,y}dy
      \]
      and 
    \[\hat{R}_{x,\pi}(a) = \int_{\mathcal{Y}}\hat{\Gamma}_{y}^{\pi_{n-1}}(a)B_{x,y}dy \]
Using $\Gamma^{\pi}(a) = \sum_{i=1}^{\NumAct} \hat{\Gamma}^{\pi}(i)M(i,a)$ (where $\NumAct$ is the cardinality of the action space) yields
\begin{align}
    \begin{split}
        &R_{x,\pi}(a) = \int_{\mathcal{Y}}\sum_{i=1}^{\NumAct} \hat{\Gamma}_{y}^{\pi_{n-1}}(i)M(i,a)B_{x,y}dy \\ 
        &= \sum_{i=1}^{\NumAct}\int_{\mathcal{Y}}\hat{\Gamma}_{y}^{\pi_{n-1}}(i)B_{x,y}dy M(i,a) \\
        &= \sum_{i=1}^{\NumAct} \hat{R}_{x,\pi}(i) M(i,a)
    \end{split}
\end{align}
Now, following \eqref{glob_dec_mak}:
\begin{align}
    \begin{split}
        &{T}(\pi,a) = \frac{R_{\pi}(a)P'\pi}{{\sigma}(\pi,a)}, \ {\sigma}(\pi,a) = \boldsymbol{1}'R_{\pi}(a)P'\pi \\
        &\hat{T}(\pi,a) = \frac{\hat{R}_{\pi}(a)P'\pi}{\hat{\sigma}(\pi,a)}, \ \hat{\sigma}(\pi,a) = \boldsymbol{1}'\hat{R}_{\pi}(a)P'\pi \\
        & R_{\pi}(a) = \textrm{diag}(R_{1,\pi}(a), R_{2,\pi}(a)) \\ 
        & \hat{R}_{\pi}(a) = \textrm{diag}(\hat{R}_{1,\pi}(a), \hat{R}_{2,\pi}(a))
    \end{split}
\end{align}
Now observe that we can manipulate $T(\pi,a)$ in the following way:
\begin{align}
    \begin{split}
    \label{T_form}
        &T(\pi,a) = \frac{R_{\pi}(a)P'\pi}{{\sigma}(\pi,a)} = \frac{\sum_{i=1}^{\NumAct} \hat{R}_{\pi}(i)M(i,a) P'\pi}{{\sigma}(\pi,a)} \\
        &= \frac{\sum_{i=1}^{\NumAct} \hat{R}_{\pi}(i)P'\pi M(i,a)}{{\hat{\sigma}}(\pi,a)}
        \frac{{\hat{\sigma}}(\pi,a)}{{\sigma}(\pi,a)} \\
        &= \sum_{i=1}^{\NumAct}\hat{T}(\pi,i)\frac{\hat{\sigma}(\pi,i)}{\sigma(\pi,a)}M(i,a)
    \end{split}
\end{align}
We now use induction in the value iteration algorithm \eqref{value_itr}. The algorithm begins with $\mathcal{V}_0(\pi) = \hat{\mathcal{V}}_0(\pi) = 0 \ \forall \pi \in \Pi$, so we trivially have $\mathcal{V}_0(\pi) \geq \hat{\mathcal{V}}_0(\pi)$. 
We also know that the value function $\hat{\mathcal{V}}_k$ is concave for all $\pi \in \Pi$ (see proof of Theorem \ref{thm:sing_thres}), so Jensen's inequality can be invoked to produce
\begin{align}
    \begin{split}
        &\hat{\mathcal{V}}(T(\pi,a)) = \hat{\mathcal{V}}\left(\sum_{i=1}^{\NumAct}\hat{T}(\pi,i)\frac{\hat{\sigma}(\pi,i)}{\sigma(\pi,a)}M(i,a)\right) \\
        &\geq \sum_{i=1}^{\NumAct}\hat{\mathcal{V}}(\hat{T}(\pi,i))\frac{\hat{\sigma}(\pi,i)}{\sigma(\pi,a)}M(i,a)
    \end{split}
\end{align}
Thus we get
\begin{align}
    \begin{split}
        \sum_{a=1}^{\NumAct}\hat{\mathcal{V}}(T(\pi,a))\sigma(\pi,a) \geq \sum_{a=1}^{\NumAct}\hat{\mathcal{V}}(\hat{T}(\pi,a)){\hat{\sigma}}(\pi,a)
    \end{split}
\end{align}

and, assuming $\mathcal{V}_k(\pi) \geq \hat{\mathcal{V}}_k(\pi)$, we have
\begin{align}
    \begin{split}
        &C(\pi,2) + \sum_{a=1}^{\NumAct}\mathcal{V}(T(\pi,a))\sigma(\pi,a) \\
        &\geq C(\pi,2) + \sum_{a=1}^{\NumAct}\hat{\mathcal{V}}(T(\pi,a))\sigma(\pi,a) \\
        &\geq C(\pi,2) + \sum_{a=1}^{\NumAct}\hat{\mathcal{V}}(\hat{T}(\pi,a))\hat{\sigma}(\pi,a)
    \end{split}
\end{align}
Thus $\mathcal{V}_{k+1}(\pi) \geq \hat{\mathcal{V}}_{k+1}(\pi)$ and the induction step is complete. The value iteration algorithm \eqref{value_itr} converges pointwise, so $\mathcal{V}(\pi) \geq \hat{\mathcal{V}}(\pi)$.

\end{proof}

\subsection{Proof of Theorem \ref{thm:convex_dom}}
\label{ap:B}
\begin{proof}
We prove the existence of such a stochastic matrix by construction. First, we have that $\Gamma_1 = \hat{\Gamma}M_1$ and $\Gamma_2 = M_2\hat{\Gamma}$, so
\[\Gamma_1(a) = \sum_{i=1}^N \hat{\Gamma}(i)M_1(i,a), \ \ \  \Gamma_2(a) = \sum_{i=1}^N M_2(a,i)\hat{\Gamma}(i)M_2(i,1)\]
Then 
\begin{align}
    \begin{split}
       & \Gamma_3(a) = \gamma_a\Gamma_1(a) + (1-\gamma_a)\Gamma_2(a) \\
       &= \gamma_a\sum_{i=1}^N \hat{\Gamma}(i)M_1(i,a) + (1-\gamma_a)\sum_{i=1}^N \hat{\Gamma}(i)M_2(i,a) \\
       &= \sum_{i=1}^N \left(\gamma_a M_1(i,a) + (1-\gamma_a)M_2(i,a)\right)\hat{\Gamma}(i)
    \end{split}
\end{align}
Now simply form matrix $M_3$ as 
\[M_3(i,a) =  \gamma_a M_1(i,a) + (1-\gamma_a)M_2(i,a)\]
so that $\Gamma_3 = M_3\hat{\Gamma}$. It is also easily verified that $M_3$ is stochastic, since $M_1$ and $M_2$ are stochastic:
\[\sum_{i=1}^N \gamma_a M_1(i,a) + (1-\gamma_a)M_2(i,a) = 1 \ \forall a \in \{1,\dots,N\}\]
\end{proof}


\ifCLASSOPTIONcaptionsoff
  \newpage
\fi

\bibliographystyle{IEEEtran}
\bibliography{Bibliography.bib}

%








\end{document}